\documentclass[pdflatex,sn-mathphys-num]{sn-jnl}

\usepackage{multirow}%
\usepackage{amsmath,amssymb,amsfonts}%
\usepackage{amsthm}%
\usepackage{mathrsfs}%
\usepackage[title]{appendix}%
\usepackage{xcolor}%
\usepackage{textcomp}%
\usepackage{manyfoot}%
\usepackage{booktabs}%
\usepackage{algorithm}%
\usepackage{algorithmicx}%
\usepackage{algpseudocode}%
\usepackage{listings}%
\usepackage{enumitem}
\usepackage{bm}
\usepackage{float}%
\usepackage{siunitx}
\usepackage{tabularx}
\usepackage{ltablex}

\newcommand\setrow[1]{\gdef\rowmac{#1}#1\ignorespaces}
\newcommand\clearrow{\global\let\rowmac\relax}
\clearrow

\newtheorem{theorem}{Theorem}
\newtheorem{lemma}{Lemma}
\newtheorem{corollary}{Corollary}
\newtheorem{example}{Example}
\newtheorem{definition}{Definition}

\begin{document}
		
		\title{A sequence-form characterization and differentiable path-following method for computing normal-form perfect equilibria in extensive-form games}
		
		\author[1,3]{\fnm{Yuqing} \sur{Hou}}\email{yuqinghou2-c@my.cityu.edu.hk}
		
		\author*[2]{\fnm{Yiyin} \sur{Cao}}\email{yiyincao2-c@my.cityu.edu.hk}
		
		\author[3]{\fnm{Chuangyin} \sur{Dang}}\email{mecdang@cityu.edu.hk}
		
		\author[1]{\fnm{Yong} \sur{Wang}}\email{yongwang@ustc.edu.cn}
		
		\affil[1]{\orgdiv{Department of Automation}, \orgname{University of Science and Technology of China}, \orgaddress{\city{Hefei}, \country{China}}}
		
		\affil[2]{\orgdiv{School of Management}, \orgname{Xi'an Jiaotong University}, \orgaddress{\city{Xi'an}, \country{China}}}
		
		\affil[3]{\orgdiv{Department of Systems Engineering}, \orgname{City University of Hong Kong}, \orgaddress{\city{Hong Kong}, \country{China}}}
		
		\abstract{
			The sequence form, owing to its compact and holistic strategy representation, has demonstrated significant efficiency in computing normal-form perfect equilibria for two-player extensive-form games with perfect recall. Nevertheless, the examination of $n$-player games remains underexplored. To tackle this challenge, we present a sequence-form characterization of normal-form perfect equilibria for $n$-player extensive-form games, achieved through a class of perturbed games formulated in sequence form. Based on this characterization, we develop a differentiable path-following method for computing normal-form perfect equilibria and prove its convergence. This method formulates an artificial logarithmic-barrier game in sequence form, introducing an additional variable to regulate the impact of logarithmic-barrier terms on the payoff functions, as well as the transition of the strategy space. We prove the existence of a smooth equilibrium path defined by the artificial game, starting from an arbitrary positive realization plan and converging to a normal-form perfect equilibrium of the original game as the additional variable approaches zero. Furthermore, we extend Harsanyi's linear and logarithmic tracing procedures to the sequence form and develop two alternative methods for computing normal-form perfect equilibria. Numerical experiments further substantiate the effectiveness and computational efficiency of our methods.
		}
		
			
		
	\maketitle
	
	
	\section{Introduction}
	Game theory provides a comprehensive mathematical framework for decision optimization in settings involving strategic interaction among rational agents. The primary concerns within this field revolve around the representation and resolution of games. The extensive-form game~\cite{KuhnExtensiveGames1950} provides a significant representation, especially applicable in scenarios involving sequential interactions. As a central solution concept for extensive-form games, Nash equilibrium signifies a state in which no player can improve their payoff by changing their strategy alone. Nevertheless, as explored by Selten~\cite{SeltenReexaminationperfectnessconcept1975}, Myerson~\cite{MyersonRefinementsNashequilibrium1978}, Kreps and Wilson~\cite{KrepsSequentialEquilibria1982}, and van Damme~\cite{vanDammeStabilityperfectionNash1987}, an extensive-form game can have multiple Nash equilibria, some of which may deviate from our intuitive expectations regarding the game's outcome and off-equilibrium strategies. As a refinement of Nash equilibrium, perfect equilibrium introduced by Selten~\cite{SeltenReexaminationperfectnessconcept1975} can eliminate a variety of counter-intuitive equilibria by introducing slight perturbations to the strategies. According to Selten, the concept of perfect equilibrium in extensive-form games can be classified into two types: extensive-form perfect equilibrium and normal-form perfect equilibrium, neither of which is contained within the other. As noted by van Damme~\cite{vandammeRelationPerfectEquilibria1984}, extensive-form perfect equilibria may involve dominated strategies. In addition, Kohlberg and Mertens~\cite{KohlbergStrategicStabilityEquilibria1986} contended that the reduced normal form contains all essential information required for decision-making. Stalnaker~\cite{stalnakerExtensiveStrategicForms1999} further emphasized the sufficiency of normal-form representations in the context of epistemic models. Therefore, it is crucial to study the equilibrium refinements of extensive-form games in normal form. This paper investigates the computation of normal-form perfect equilibria in finite $n$-player extensive-form games with perfect recall.
	
	The typical methods for computing normal-form perfect equilibria in extensive-form games rely on transforming these games into their normal-form representations, followed by the application of equilibrium computation methods specific to normal-form games. Notably, path-following methods have gained prominence as powerful and effective tools for computing equilibria in normal-form games. These methods are fundamentally grounded in the computation of Nash equilibria. Lemke and Howson~\cite{LemkeEquilibriumPointsBimatrix1964} proposed a complementarity pivoting algorithm to obtain Nash equilibria for bimatrix games, and later, various other algorithms in~\cite{vandenelzenProcedureFindingNash1991, aliExploringTwoNew2024, hussainNumericalExplorationTwo2025} for bimatrix games have been developed. Subsequent extensions by Rosenmüller~\cite{RosenmullerGeneralizationLemkeHowson1971} and Wilson~\cite{WilsonComputingEquilibriaNPerson1971} independently adapted the Lemke-Howson algorithm to $n$-player games. To render this extended method computationally feasible, Garcia et al.~\cite{GarciaSimplicialApproximationEquilibrium1973} introduced a simplicial path-following method and implemented it to approximate Nash equilibria. Over the subsequent decades, significant advancements were made in the development of simplicial path-following methods for computing Nash equilibria in normal-form games. Noteworthy contributions include the work of van der Laan and Talman~\cite{vanderLaanComputationFixedPoints1982}, Doup and Talman~\cite{Doupnewsimplicialvariable1987}, and Herings and van den Elzen~\cite{HeringsComputationNashEquilibrium2002}, who proposed increasingly flexible and efficient methods. Although these methods are capable of ultimately reaching Nash equilibria, their convergence rates are notably hindered by the failure to exploit the differentiability inherent in games~\cite{allgowerPiecewiseLinearMethods2000a}. In response, several differentiable path-following methods have been presented in the literature, with key contributions from Herings and Peeters~\cite{Heringsdifferentiablehomotopycompute2001}, Harsanyi and Selten ~\cite{HarsanyiGeneralTheoryEquilibrium1988}, Govindan and Wilson~\cite{GovindanglobalNewtonmethod2003}, and Chen and Dang~\cite{Chenreformulationbasedsmoothpathfollowing2016}. These approaches have markedly enhanced the convergence rates when computing Nash equilibria in normal-form games.
	
	Research into the computation of Nash equilibria in normal-form games has led to the development of various methods for determining perfect equilibria. van den Elzen and Talman~\cite{vandenelzenProcedureFindingNash1991} developed the first method to compute a perfect equilibrium, employing a complementary pivoting algorithm that operates exclusively in bimatrix games. Chen and Dang \cite{chenReformulationBasedSimplicialHomotopy2019} generalized Kohlberg and Mertens's Nash equilibrium reformulation~\cite{KohlbergStrategicStabilityEquilibria1986, Chenreformulationbasedsmoothpathfollowing2016} to perturbed games and proposed a simplicial path-following method for identifying a perfect equilibrium in $n$-person games. Subsequently, Chen and Dang \cite{Chenextensionquantalresponse2020} extended the logistic quantal response equilibrium, demonstrating the existence of a smooth path to a perfect equilibrium under a specific assumption regarding payoff functions. Chen and Dang~\cite{Chendifferentiablehomotopymethod2021} introduced an exterior-point differentiable homotopy method capable of selecting an approximate perfect equilibrium. Motivated by the limitations of prior methods and guided by the principles of Harsanyi and Selten's equilibrium selection philosophy~\cite{HarsanyiGeneralTheoryEquilibrium1988}, more alternative schemes for computing perfect equilibria have been developed. By exploiting the selection properties inherent in the Nash's mappings, Harsanyi's tracing procedures, and logistic quantal response equilibrium, Cao and Dang~\cite{caoComplementarityEnhancedNashs2022,CaovariantHarsanyitracing2022,caoVariantLogisticQuantal2024} developed their respective variants and distinct differentiable path-following methods to select an exact perfect equilibrium in normal-form games. Differentiable path-following methods have also shown remarkable performance in computing other refinements of Nash equilibria~\cite{Zhansmoothpathfollowingalgorithm2018, caoDifferentiablePathFollowingMethod2023}.
	
	While several differentiable path-following methods exist for computing perfect equilibria in normal-form games, the exponential growth inherent in normal-form representations significantly undermines computational tractability, rendering these methods impractical even for games of moderate scale. Although reduced normal-form representations achieve notable size reductions, the fundamental problem of exponential expansion persists~\cite{dalkey1953equivalence}. To mitigate the complexity arising from this transformation, Wilson~\cite{wilsonComputingEquilibriaTwoPerson1972} and Koller and Megiddo~\cite{KollerFindingmixedstrategies1996} suggested employing mixed strategies with small supports in two-player extensive-form games. Furthermore, Koller and Megiddo~\cite{Kollercomplexitytwopersonzerosum1992} developed a polynomial-time algorithm that uses realization weights at nodes to solve two-person zero-sum extensive-form games with perfect recall, despite the exponential increase in constraints. To overcome these scalability issues, von Stengel~\cite{vonStengelEfficientComputationBehavior1996} introduced the sequence form, a more compact representation that reformulates Nash equilibrium computation with improved efficiency. In this representation, pure strategies are replaced by sequences and random strategies comply with a recursively defined linear system. Building upon this, substantial progress has been made for two-player extensive-form games with perfect recall. Koller et al.~\cite{KollerEfficientComputationEquilibria1996} developed an algorithm that applies Lemke's method~\cite{LemkeBimatrixEquilibriumPoints1965} to the linear complementarity problem derived from the sequence form. Its practical efficiency was demonstrated through the Gala system by Koller and Pfeffer~\cite{KollerRepresentationssolutionsgametheoretic1997}. Subsequently, von Stengel et al.~\cite{VonStengelComputingNormalForm2002} extended van den Elzen and Talman's~\cite{vandenelzenProcedureFindingNash1991} method for computing perfect equilibria in normal-form games to the sequence form, enabling the computation of normal-form perfect equilibria in extensive-form games. Miltersen and Sørensen~\cite{MiltersenComputingquasiperfectequilibrium2010} further adapted Koller et al.'s~\cite{KollerEfficientComputationEquilibria1996} approach to handle perturbed games, facilitating the computation of quasi-perfect equilibria. Research on games involving more than two players remains limited. Govindan and Wilson~\cite{govindanStructureTheoremsGame2002} extended structure theorems for perturbed extensive-form games by introducing strategies in sequence form, leading to a piecewise differentiable path-following method for computing Nash equilibria in $n$-player perturbed extensive-form games. Hou et al.~\cite{houSequenceformDifferentiablePathfollowing2025} proposed a globally differentiable path-following method within the sequence form framework to compute Nash equilibria. However, no sequence-form globally differentiable path-following methods currently exist to compute normal-form equilibrium refinements, including normal-form perfect equilibria. A major barrier to further progress is the lack of a theoretical foundation for characterizing normal-form perfect equilibria in sequence form. Although Gatti et al.~\cite{gattiCharacterizationQuasiperfectEquilibria2020} presented a sequence-form characterization for quasi-perfect equilibria, their reliance on sequential structures conflicts with the simultaneity principle inherent in the normal form, making it insufficient for fully capturing normal-form perfect equilibria.
	
	The objective of this study is to develop a sequence-form characterization of normal-form perfect equilibrium in $n$-player extensive-form games with perfect recall and, based on this, to propose an effective and efficient globally differentiable path-following method for its computation. To achieve this, we begin by establishing the equivalence relationship between strategies in normal-form and sequence-form representations, as well as their connection to best-response strategies. Following this, we introduce a class of perturbed games in sequence form and apply an optimization-based approach to derive the corresponding equilibrium systems for these perturbed games. We then demonstrate the necessity and sufficiency of the solution limits of these equilibrium systems for characterizing normal-form perfect equilibria. To develop a differentiable path-following method based on this characterization, we construct an artificial game in sequence form by introducing an additional variable and embedding logarithmic barrier terms within the payoff functions. As the additional variable decreases from two to zero, the artificial game experiences two distinct transformation phases. The initial phase aims to identify a unique starting point by shifting strategies from the constructed form to realization plans. In the subsequent phase, the logarithmic barrier terms confine the strategies within the interior of the realization plan space, thus guaranteeing a well-defined mapping from realization plans to mixed strategies. We establish the existence of a smooth equilibrium path governed by the artificial game, which converges to a normal-form perfect equilibrium of the original game as the additional variable tends towards zero. Lastly, numerical experiments are performed to demonstrate the effectiveness and computational efficiency of our proposed methods.
	
	The rest of this paper is structured as follows. Section~\ref{nfpe-sec-prm1} reviews the concept of normal-form perfect equilibrium within extensive-form games and introduces the sequence form. In Section~\ref{nfpe-sec-prm2}, we provide a sequence-form characterization of normal-form perfect equilibria. Section~\ref{nfpe-sec-prm3} proposes a globally differentiable logarithmic path-following method based on the sequence form for computing normal-form perfect equilibria. In Section~\ref{nfpe-sec-prm5}, we extend Harsanyi's linear and logarithmic tracing procedures to the sequence form, resulting in two alternative computational approaches. Numerical results and comparative evaluations of these methods are presented in Section~\ref{nfpe-sec-prm6}, with concluding remarks given in Section~\ref{nfpe-sec-prm7}.
	\section{Preliminaries}\label{nfpe-sec-prm1}
	\begin{table}[tb!]
		\centering
		\caption{Notation for Extensive-Form Games}
		\begin{tabular}{ll}
			\toprule
			Symbol & Terminology\\
			\midrule
			$N=\{1,2,\ldots,n\}$ & Set of players\\
			$N_c=N\cup\{c\}$ & Set of players and chance player $c$\\
            $-i$ & All non-chance players excluding player $i\in N$\\
			$a$ & Action taken by a player\\
			$H$ & Set of histories,  \\
			&$\emptyset\in H$ and $\langle a_1,\ldots,a_L\rangle\in H$ if $\langle a_1,\ldots,a_K\rangle\in H$ and $L<K$\\
			$Z$ & Set of terminal histories\\
			$A(h)=\{a:(h,a)\in H\}$ & Set of actions after a nonterminal history $h$\\
            $R_{i}(h)$ & Record of player $i\in N_c$'s experience along $h$\\
			$P(h)$ & Player who takes an action after $h$\\
			$f_{c}(a|h)$ & Probability that chance player $c$ takes action $a$ after $h$\\
			$\mathcal{I}_{i}$ & Collection of information partitions of $\{h\in H|P(h)=i\}$\\
			$M_{i}=\{1,\ldots,m_{i}\}$ & Set of information partition indices for player $i\in N_c$\\
			$I^{j}_{i}\in\mathcal{I}_{i},j\in M_{i}$ & $j$th information set of player $i\in N_c$, \\&$A(I^j_i)\triangleq A(h)= A(h')$ whenever $h,h'\in I^j_i$\\
			$\succsim_i$ & Preference relation of player $i\in N$ \\
			$u_z^{i}:Z\to\mathbb{R}$ & Payoff function of player $i\in N$\\
			$|C|$ & Cardinality of a finite set $C$\\
			$m_0=\sum_{i\in N}m_i$ & Number of information sets\\
			$n_0=\sum_{i\in N}\sum_{j\in M_i}|A(I^j_i)|$ & Number of actions for non-chance players\\
			$\text{int}(C)$ & Interior of the set $C$\\
			\bottomrule
		\end{tabular}
		\label{nfpe-tab-pre1}
	\end{table}
	We adopt the notation for extensive-form games from Osborne and Rubinstein~\cite{OsborneCourseGameTheory1994}, summarized in Table~\ref{nfpe-tab-pre1}. An extensive-form game can be represented by
	\[\Gamma=\langle N, H, P, f_c, \{{\cal I}_i\}_{i\in N}, \{\succsim_i\}_{i\in N}\rangle.\] This paper addresses finite extensive-form games with perfect recall. ``finite" means that $H$ is a finite set. Perfect recall holds if, for each player $i$, any histories $h$ and $h'$ in the same information set satisfy $R_i(h)=R_i(h')$, ensuring consistent memory of past actions and knowledge.
	
    The equilibrium concept we aim to investigate is the normal-form perfect equilibrium. With this in mind, we need to introduce the normal-form representation of extensive-form games. Given an extensive-form game $\Gamma$, a pure strategy $s^i$ of player $i\in N_c$ is defined as a function that maps each information set $I^j_i,j\in M_i$ to an action $a\in A(I^j_i)$. To facilitate computations, we define 
    	\begin{equation*}\label{nfpe-equ-pre0}
    	s^i(a) = \left\{\begin{array}{ll}
    		1  & \text{if $s^i(I^j_i)=a$},\\
            
    		0 & \text{otherwise.}
    	\end{array}\right.\end{equation*}
    The payoff function for player $i\in N$ under any pure strategy combination $s=\{s^i:i\in N_c\}$ is defined as 
    \begin{equation}\label{nfpe-equ-pre1}\begin{array}{l}u^i(s)=\sum\limits_{h=\langle a_1,\ldots,a_L\rangle\in Z}u^i_z(h)\prod\limits_{q=0}^{L-1}s^{P(\langle a_1,\ldots,a_q\rangle)}(a_{q+1}),\end{array}\end{equation}
    The chance player's mixed strategy $\sigma^c=(\sigma^c(s^c):s^c\in S^c)$ is fixed and determined by $\sigma^c(s^c)=\prod_{h\in H,P(h)=c}\sum_{a\in A(h)} s^c(a)f_c(a|h)$. Additional notations and their descriptions are provided in Table~\ref{nfpe-tab-pre2}. Then the normal-form representation of $\Gamma$ is expressed as $\Gamma_n=\langle N, S, \Xi, \{u^i\}_{i\in N}\rangle$.
    
    In the reduced normal-form representation, pure strategies are defined in a more compact manner while preserving all valid strategic information. Specifically, for a pure strategy $s^i$ of player $i\in N_c$, $s^i(I^j_i)=a$, $j\in M_i,a\in A(I^j_i)$ means that, for $\langle a_1,\ldots,a_L\rangle\in I^j_i$, $s^i(I^{j_q}_i)=a_q$ holds for all $0\leq q\leq L$ with $j_q\in M_i,a_q\in A(I^{j_q}_i)$. All other definitions remain unchanged and are still applicable. To highlight the superiority of our methods, all derivations in this paper are based on the reduced normal form. For simplicity, we shall refer to it as the normal form throughout, omitting the qualifier ``reduced".
    
    Given a mixed strategy profile $\sigma=(\sigma^i:i\in N)\in \Xi$, the expected payoff of player $i\in N$ is given by $u^i(\sigma)=\sum_{s^i\in S^i}\sigma^i(s^i)u^i(s^i,\sigma^{-i})$ with
    \begin{equation}\label{nfpe-equ-pre2}\begin{array}{l}u^i(s^i,\sigma^{-i})=\sum\limits_{s^{-i}\in S^{-i}}u^i(s^i,s^{-i})\prod\limits_{i_q\in N_c\backslash \{i\}}\sigma^{i_q}(s^{i_q}).\end{array}\end{equation}
    A mixed strategy profile $\sigma^*$ is referred to a Nash equilibrium if, for every player $i\in N$, the inequality $u^i(\sigma^*)\geq u^i(\sigma^i,\sigma^{*-i})$ is satisfied for all $\sigma^i\in \Xi^i$. This condition ensures that no player can improve their payoff by unilaterally deviating from their strategy in the equilibrium profile. However, the weakness of this condition can lead to a large equilibrium set, leading to the emergence of numerous counterintuitive equilibria and great uncertainty in determining which equilibrium to choose. In response to this limitation, Selten~\cite{SeltenReexaminationperfectnessconcept1975} introduced the concept of perfect equilibrium, eliminating a large number of unreasonable equilibria. The definition of normal-form perfect equilibrium in an extensive-form game is as follows.
\begin{definition}\label{nfpe-def-pre1}
    {\em Let $\Gamma$ be an extensive form game. For any sufficiently small $\varepsilon > 0$, a totally mixed strategy profile $\sigma(\varepsilon)\in \Xi$ is an $\varepsilon$-normal-form perfect equilibrium of $\Gamma$ if $\sigma^i(\varepsilon;s^i)\leq\varepsilon$ whenever $u^i(s^i, \sigma^{-i}(\varepsilon)) < u^i(\tilde s^{i},\sigma^{-i}(\varepsilon))$ for all $i\in N$ and $s^i,\tilde s^i\in S^i$. A mixed strategy profile $\sigma^*\in \Xi$ is defined as a normal-form perfect equilibrium of game $\Gamma$ if $\sigma^*$ is a limit point of some sequence $\{\sigma(\varepsilon^k)\}_{k=1}^\infty$, where $\lim_{k\to\infty}\varepsilon^k=0$ and each $\sigma(\varepsilon^k)$ is an $\varepsilon^k$-normal form perfect equilibrium of $\Gamma$.}
\end{definition}
\begin{table}[tb!]
	\centering
	\caption{Notation for Games in Normal Form or Sequence Form}
	\begin{tabular}{ll}
		\toprule
		Symbol & Terminology\\
		\midrule
		$s^i$ & Pure strategy of player $i$\\
		$S=\underset{i\in N_c}{\times}S^i$ & Set of pure strategy profiles\\
		$\sigma^i$ & Mixed strategy of player $i\in N_c$, probability measure over $S^i$\\
		$\Xi=\underset{i\in N}{\times}\Xi^i$ & Set of mixed strategy profiles,  $\Xi^i=\{\sigma^i\in\mathbb{R}_+^{|S^i|}|\sum\limits_{s^i\in S^i}\sigma^{i}(s^i)=1\}$\\
		$\text{int}(\Xi)=\underset{i\in N}{\times} \text{int}(\Xi^i)$ & Set of totally mixed strategy profiles\\
		$u^i(s)$ & Expected payoff of player $i$ on the pure strategy profile $s$\\
		$\varpi^i$ & Sequence of actions taken by player $i$\\
		$\varpi^i_{I^j_i}$ & Sequence of player $i$ leading to $I^j_i$, $\varpi^i_h=\varpi^i_{I^j_i}$ for any $h\in I^j_i$ \\
		$\varpi^i_{I^j_i}a$ & Extended sequence $\varpi^i_{I^j_i}\cup\{a\}$\\
		${W}=\underset{i\in N_c}\times{W}^i$ & Collection of sequence profiles, $\emptyset\in{W}^i$\\
		$g^i(\varpi)$ & Expected payoff of player $i$ on the sequence profile $\varpi$\\
		$\gamma^i$ & Realization plan of player $i\in N_c$\\
		$\Lambda=\underset{i\in N}\times{ \Lambda^i}$ & Set of realization plan profiles\\
		$M_i(\varpi^i)$ & Set of $j\in M_i$ satisfying $\varpi^i_{I^j_i}=\varpi^i$\\
		$D_i$ & Set of $(j,a)$ for player $i$ with $M_i(\varpi^i_{I^j_i}a)=\emptyset$\\
		\bottomrule
	\end{tabular}
	\label{nfpe-tab-pre2}
\end{table}
The computation of a normal-form perfect equilibrium of an extensive-form game typically requires a transformation into its normal form. As Wilson~\cite{wilsonComputingEquilibriaTwoPerson1972} points out, even simple extensive-form games often produce exceedingly large normal forms due to the exponential increase in the number of pure strategies relative to the number of information sets. To circumvent this exponential growth, the sequence form, formally developed by von Stengel~\cite{vonStengelEfficientComputationBehavior1996}, has emerged as a particularly efficient alternative.

The sequence form replaces pure strategies with sequences, providing a compact representation. For $i\in N_c$, a sequence $\varpi^i$ is defined as the action set of player $i$ for some history. Specifically, for $h=\langle a_1,\ldots,a_L\rangle\in H$, the corresponding sequence is given by\[\varpi^i_h=\{a_q:\text{$a_q\in A(I^j_i)$ for some $j\in M_i$ and $1\leq q\leq L$}\},\]which is either an empty sequence $\emptyset$ or an extension $\varpi^i_{h'}a$ of a preceding sequence $\varpi^i_{h'}$ with $i \in N$ and $h'\in H$. The function $g^i$ specifies the payoff of player player $i$ for any sequence profile $\varpi\in{W}$, defined by
	\begin{equation*}\label{rplanpayoff}
		g^i(\varpi) = \left\{\begin{array}{ll}
			u^i_z(h)  & \text{if $\varpi$ is defined by $h\in Z$,}\\
			
			0 & \text{otherwise.}
		\end{array}\right.\end{equation*}
We say that $\varpi=(\varpi^i:i\in N_c) \in {W}$ is defined by $h=\langle a_1,\ldots,a_L\rangle$ if $\mathop{\cup}_{i\in N_c} \varpi^i=\{a_1,\ldots,a_L\}$.
	
	The inherent challenge in developing algorithms for the sequence form lies in the fact that randomization over sequences is no longer represented by a probability distribution over the strategies at an information set. Instead, it requires the formulation of an recursive system of linear equations. For player $i\in N_c$, a random strategy in the sequence form is a function $\gamma^i$ defined on ${W}^i$, with the convention that $\gamma^i(\emptyset) = 1$. We call $\gamma^i$ a realization plan for player $i$ if it satisfies the linear system,
	\begin{equation}\label{nfpe-equ-pre3}\begin{array}{l}
			\sum\limits_{a\in A(I^j_i)}\gamma^i(\varpi^i_{I^j_i}a)-\gamma^i(\varpi^i_{I^j_i})=0,\;j\in M_i,\\
			0\le \gamma^i(\varpi^i_{I^j_i}a),\;j\in M_i,a\in A(I^j_i).
		\end{array}
	\end{equation}	
	This recursive system~(\ref{nfpe-equ-pre3}) suggests that the realization plan $\gamma^i$ is uniquely determined by the values of $\gamma^i(\varpi^i_{I^j_i}a),(j,a)\in D_i$, where no information set of player $i$ succeeds $\varpi^i_{I^j_i}a$. It reflects the holistic property of the sequence form. The chance player's realization plan $\gamma^c=(\gamma^c(\varpi^c):\varpi^c\in W^c)$ is determined by $\gamma^c(\varpi^c)=\prod_{a\in \varpi^c\cap A(h)}f_c(a|h)$, which satisfies the system~(\ref{nfpe-equ-pre3}). More notations and their descriptions can be found in Table~\ref{nfpe-tab-pre2}. The sequence form of an extensive-form game is represented as \[\Gamma_s=\langle N,\{{W}^i\}_{i\in N_c},\gamma^c,\{g^i\}_{i\in N}\rangle.\] Given a realization plan profile $\gamma=(\gamma^i:i\in N_c)$, the expected payoff for player $i\in N$ at sequence $\varpi^i\in W^i$ is defined as 
	\[\begin{array}{l}g^i(\varpi^i,\gamma^{-i})=\sum\limits_{\varpi^{-i}\in {W}^{-i}}g^i(\varpi^i,\varpi^{-i})\prod\limits_{i_q\ne i}\gamma^{i_q}(\varpi^{i_q}).\end{array}\]
	Thus, the overall expected payoff for player $i\in N$ can be written as
	\[\begin{array}{l}g^i(\gamma)=\sum\limits_{\varpi^i\in W^i}\gamma^i(\varpi^i)g^i(\varpi^i,\gamma^{-i}).\end{array}\]
	The number of sequences available to player $i$ is given by $\sum_{j \in M_i} |A(I_i^j)|+1$, exhibiting a linear growth relative to the number of information sets. This compactness, in conjunction with holism, makes the sequence form a crucial framework for developing efficient methods to compute normal-form perfect equilibria in extensive-form games.
	
	\section{A Sequence-Form Characterization of Normal-Form Perfect Equilibria}\label{nfpe-sec-prm2}
	This section begins by examining the relationship between mixed strategies and realization plans, laying the groundwork for the characterization of normal-form perfect equilibria. Following this, a sequence-form characterization of normal-form perfect equilibria are established as a limit point of a sequence of Nash equilibria within a class of perturbed games.
	\subsection{Relationship between Mixed Strategies and Realization Plans}
	von Stengel et al.~\cite{VonStengelComputingNormalForm2002} have conducted an preliminary exploration into the payoff equivalence between mixed strategies and realization plans. In this study, we offer a more detailed description and a rigorous proof of this relationship. Furthermore, we expand the analysis by incorporating an examination of best response strategies.
	
	Consider an extensive-form game, $\Gamma$, with $\Gamma_n$ as its normal form and $\Gamma_s$ as its sequence form. Given any pure strategy $s^i\in S^i$ of player $i\in N_c$, we define $s^i(\varpi^i)=\prod_{a\in \varpi^i}s^i(a)$ for $\varpi^i\in W^i$. For any $\sigma\in \Xi$, let $\gamma(\sigma)=(\gamma^{i}(\sigma^i;\varpi^i):i\in N_c,\varpi^i\in W^i)$, where $\gamma^{i}(\sigma^i;\varpi^i)=\sum_{s^i\in S^i}s^i(\varpi^i)\sigma^i(s^i)$. This entails that $\gamma^{i}(s^i;\varpi^i)=s^i(\varpi^i)$ and $\gamma(\sigma)\in\Lambda$. This relation is captured by the set $T=\left\{(\sigma,\gamma)\left|\sigma\in  \Xi,\gamma=\gamma(\sigma)\right.\right\}$, which leads to the following lemma. 
	\begin{lemma}\label{nfpe-lem-sfc1}{\em For any $\gamma\in\Lambda$, there exists a mixed strategy profile $\sigma$ such that $(\sigma,\gamma)\in T$.}\end{lemma}
	\begin{proof}
		For $\gamma\in\text{int}(\Lambda)$, define $\sigma(\gamma)=(\sigma^i(\gamma^i;s^i):i\in N_c,s^i\in S^i),
		\text{where }\sigma^i(\gamma^i;s^i)=\prod_{j\in M_i,a\in A(I^j_i),s^i(a)=1}\gamma^i(\varpi^i_{I^i_j}a)/\gamma^i(\varpi^i_{I^j_i})$. Consider any $\gamma^*\in \Lambda$, and let $\{\gamma^k\in\text{int}(\Lambda)\}_{k=1}^{\infty}$ be a sequence converging to $\gamma^*$, i.e., $\lim_{k\rightarrow\infty}\gamma^k=\gamma^*$. It follows that $\sigma(\gamma^k)\in\text{int}(\Xi)$ and $(\sigma(\gamma^k),\gamma^k)\in T$ for each $k$. Because the sequence $\{(\sigma(\gamma^k),\gamma^k)\}_{k=1}^{\infty}$ lies within the compact set $ T$, there exists a convergent subsequence. Denote the limit mixed strategy of this subsequence as $\sigma^*\in  \Xi$, we have $(\sigma^*,\gamma^*)\in T$. This completes the proof.
	\end{proof}
	\begin{lemma}\label{nfpe-lem-sfc2}{\em If $(\sigma,\gamma)\in T$, $u^i(\sigma)=g^i(\gamma)$ holds for every player $i\in N$.}\end{lemma}
	\begin{proof}
		For a mixed strategy profile $\sigma\in \Xi$, the probability of reaching each terminal history $h\in Z$ is $\prod_{i\in N_c}\sum_{s^i\in S^i}s^i(\varpi^i_h)\sigma^i(s^i)$. For a realization plan profile $\gamma\in \Lambda$, it is $\prod_{i\in N_c}\gamma^{i}(\varpi^i_h)$. When $(\sigma,\gamma)\in T$, we have $\gamma^{i}(\varpi^i_h)=\sum_{s^i\in S^i}s^i(\varpi^i_h)\sigma^i(s^i)$ for every player $i\in N_c$. As a result, the probabilities that reaching each terminal history $h\in Z$ under the strategies $x$ and $\gamma$ coincide, which implies that $u^i(\sigma)=u^i(\gamma)$. This finalizes the proof.
	\end{proof}
	After establishing the payoff equivalence between the two type of strategies, we proceed to analyze the relationship between their best responses, which underpin the proof presented in the subsequent subsection. Given $\gamma\in\Lambda$, we define the expected payoff, leading by the sequence $\varpi^i\in W^i$, for player $i\in N$ as 
	\[\begin{array}{l}
		g^i(\gamma;\varpi^i)=\sum\limits_{\tilde{\varpi}^i\in W^i,\varpi^i\subseteq\tilde{\varpi}^i}\gamma^i(\tilde{\varpi}^i)g^i(\tilde{\varpi}^i,\gamma^{-i}).
	\end{array}\]
	\begin{definition}{\em
			Consider a realization plan profile $\gamma\in \Lambda$. For any $i\in N,j\in M_i,a\in A(I^j_i)$, we refer to $\varpi^i_{I^j_i}a$ as an $I^j_i$-best-response sequence to $\gamma$ if the following equality holds for any $a'\in A(I^j_i)$,
			\[\begin{array}{l}
				\max\limits_{\tilde\gamma^i\in \Lambda^i} g^i(\tilde\gamma^i,\gamma^{-i};\varpi^i_{I^j_i}a)\geq\max\limits_{\tilde\gamma^i\in \Lambda^i} g^i(\tilde\gamma^i,\gamma^{-i};\varpi^i_{I^{j}_i}a').
			\end{array}\]
			We define $\varpi^i_{I^j_i}a$ as a best-response sequence to $\gamma$ if, for any $j_q\in M_i,a_q\in A(I^{j_q}_i)$ with $a_q\in \varpi^i_{I^j_i}a$, $\varpi^i_{I^{j_q}_i}a_q$ qualifies as an $I^{j_q}_i$-best-response sequence to $\gamma$.}
	\end{definition}
	Next, we demonstrate the connection between optimal pure strategies and best-response sequences, formalized in the following lemma.
	\begin{lemma}\label{nfpe-lem-sfc3}{\em
			For $(\sigma,\gamma)\in T$ and player $i$, the following statements are equivalent:
			\begin{enumerate}[label=(\arabic*)]
				\item $u^i(s^i,\sigma^{-i})\geq u^i(\tilde s^i,\sigma^{-i})$ holds for any $\tilde s^i\in S^i$.
				\item For any $j\in M_i,a\in A(I^j_i)$ with $s^i(\varpi^i_{I^j_i}a)=1$, $\varpi^i_{I^j_i}a$ is a best-response sequence to $\gamma$.
			\end{enumerate} }
	\end{lemma}
	\begin{proof}$\bm{(1)\Rightarrow(2)}$: Assume (1) holds and (2) does not hold, there exists $j\in M_i,a\in A(I^j_i)$ with $s^i(\varpi^i_{I^j_i}a)=1$ satisfying that $\varpi^i_{I^j_i}a$ is not a best-response sequence to $\gamma$. This means that, for some $j_q\in M_i,a_q\in A(I^{j_q}_i)$  with $a_q\in \varpi^i_{I^j_i}a$, $\varpi^i_{I^{j_q}_i}a_q$ is not a $I^{j_q}_i$-best-response sequence to $\gamma$. There exists $a'_q\in A(I^{j_q}_i)$ such that
		\[\max\limits_{\tilde\gamma^i\in \Lambda^i} g^i(\tilde\gamma^i,\gamma^{-i};\varpi^i_{I^{j_q}_i}a_q)<\max\limits_{\tilde\gamma^i\in \Lambda^i} g^i(\tilde\gamma^i,\gamma^{-i};\varpi^i_{I^{j_q}_i}a'_q),\]
	which brings a pure strategy $\tilde s^i$ such that $\gamma^i(\tilde s^i)\in\arg\max_{\tilde\gamma^i\in \Lambda^i} g^i(\tilde\gamma^i,\gamma^{-i};\varpi^i_{I^{j_q}_i}a'_q)$ and $\gamma^i(\tilde s^i;\varpi^i)=\gamma^i(s^i;\varpi^i)$ for any $\varpi^i\in W^i$ with $a_q,a'_q\notin \varpi^i$. As a result, $g^i(\gamma^i(s^i),\gamma^{-i})<g^i(\gamma^i(\tilde s^i),\gamma^{-i})$, which, according to Lemma~\ref{nfpe-lem-sfc2}, implies that $u^i(s^i,\sigma^{-i})< u^i(\tilde s^i,\sigma^{-i})$, thereby resulting in a contradiction.
	
	$\bm{(2)\Rightarrow(1)}$:
		Assume $(2)$ holds, for any $j\in M_i,a\in A(I^j_i)$ with $s^i(\varpi^i_{I^{j}_i}a)=1$, we have
		\[\begin{array}{l}
			\max\limits_{\tilde\gamma^i\in \Lambda^i} g^i(\tilde\gamma^i,\gamma^{-i};\varpi^i_{I^{j}_i}a)\\
			\hspace{1.5cm}=\sum\limits_{j_q\in M_i(\varpi^i_{I^{j}_i}a)}\max\limits_{\tilde\gamma^i\in \Lambda^i}\sum\limits_{a_q\in A(I^{j_q}_i)} \tilde\gamma^i(\varpi^i_{I^{j_q}_i}a_q)g^i(\tilde\gamma^i,\gamma^{-i};\varpi^i_{I^{j_q}_i}a_q)+g^i(\varpi^i_{I^{j}_i}a,\gamma^{-i})\\
			\hspace{1.5cm}=\sum\limits_{j_q\in M_i(\varpi^i_{I^{j}_i}a)}\sum\limits_{a_q\in A(I^{j_q}_i)}s^i(\varpi^i_{I^{j_q}_i}a_q)\max\limits_{\tilde\gamma^i\in \Lambda^i} g^i(\tilde\gamma^i,\gamma^{-i};\varpi^i_{I^{j_q}_i}a_q)+g^i(\varpi^i_{I^{j}_i}a,\gamma^{-i})\\
		\end{array}
		\]
		The second equation follows directly from the condition (2). As a result of the forward induction, we can derive that
		\[\begin{array}{ll}
			\max\limits_{\tilde\gamma^i\in \Lambda^i} g^i(\tilde\gamma^i,\gamma^{-i})&=\sum\limits_{j\in M_i(\emptyset)}\max\limits_{\tilde\gamma^i\in \Lambda^i}\sum\limits_{a\in A(I^j_i)} \tilde\gamma^i(\varpi^i_{I^j_i}a)g^i(\tilde\gamma^i,\gamma^{-i};\varpi^i_{I^j_i}a)+g^i(\emptyset,\gamma^{-i})\\
			&=\sum\limits_{j\in M_i(\emptyset)}\sum\limits_{a\in A(I^j_i)}s^i(\varpi^i_{I^j_i}a)\max\limits_{\tilde\gamma^i\in \Lambda^i} g^i(\tilde\gamma^i,\gamma^{-i};\varpi^i_{I^j_i}a)+g^i(\emptyset,\gamma^{-i})\\
			&=\sum\limits_{\varpi^i\in W^i}s^i(\varpi^i) g^i(\varpi^i,\gamma^{-i})\\
			&=g^i(\gamma^i(s^i),\gamma^{-i}).
		\end{array}
		\]
		It follows that $u^i(s^i,\sigma^{-i})=\max_{\tilde \sigma^i\in \Xi^i} u^i(\tilde \sigma^i,\sigma^{-i})$. This completes the proof.
	\end{proof}
	\subsection{A Sequence-Form Characterization of Normal-Form Perfect Equilibria}
	This subsection provides a sequence-form characterization of normal-form perfect equilibria through the introduction of perturbed games in sequence form.
	
	We begin by formulating a perturbation vector. Let $\varepsilon > 0$ be a sufficiently small parameter and define a vector $\eta(\varepsilon) = (\eta^i(\varepsilon; \varpi^i) : i \in N, \varpi^i \in W^i)$, subject to the following constraints,
	\begin{equation}\label{nfpe-equ-cha1}\begin{array}{l}
			\sum\limits_{a\in A(I^j_i)}\eta^i(\varepsilon;\varpi^i_{I^j_i}a)-\eta^i(\varepsilon;\varpi^i_{I^j_i})=0,\;i\in N,j\in M_i,\\
			0< \eta^i(\varepsilon;\varpi^i_{I^j_i}a) \leq\varepsilon,\;i\in N,j\in M_i,a\in A(I^j_i).
		\end{array}
	\end{equation}
	The existence of such an $\eta(\varepsilon)$ is guaranteed by the recursiveness. Specifically, the assignment $\eta^i(\varepsilon;\varpi^i_{I^j_i}a)=\varepsilon^{|\varpi^i_{I^j_i}a|}$ for $i\in N,(j,a)\in D_i$ provides a viable solution that adheres to the conditions. Given a perturbation vector $\eta(\varepsilon)$ satisfying the system~(\ref{nfpe-equ-cha1}), let  $\Lambda(\varepsilon)=\times_{i\in N}{ \Lambda^i(\varepsilon)}$ represent the set of perturbed realization plan profiles defined by $\Lambda^i(\varepsilon)=\{\gamma^i(\varepsilon)|\gamma^i(\varepsilon)\in \Lambda^i,\gamma^i(\varepsilon;\varpi^i)\geq\eta^i(\varepsilon;\varpi^i),\varpi^i\in W^i\}$. We then construct a perturbed game in sequence form, denoted by $\Gamma_s(\varepsilon)$, where the optimal strategy for player $i$ with the strategies of other players fixed at $\hat\gamma^{-i}(\varepsilon)\in \Lambda^{-i}(\varepsilon)$ is determined by solving the linear optimization problem,
	\begin{equation}
		\label{nfpe-opt-sfc1}
		\begin{array}{rl}
			\max\limits_{\gamma^i(\varepsilon)} &\sum\limits_{j\in M_i}\sum\limits_{a\in A(I^j_i)}\gamma^i(\varepsilon;\varpi^i_{I^j_i}a)g^i(\varpi^i_{I^j_i}a,\hat\gamma^{-i}(\varepsilon))\\
			\text{s.t.}&\sum\limits_{a\in A(I^j_i)}\gamma^i(\varepsilon;\varpi^i_{I^j_i}a)-\gamma^i(\varepsilon;\varpi^i_{I^j_i})=0,\;j\in M_i,\\
			&\eta^i(\varepsilon;\varpi^i_{I^j_i}a)\le \gamma^i(\varepsilon;\varpi^i_{I^j_i}a),\;(j,a)\in D_i.
		\end{array}
	\end{equation}
	Two clarifications regarding the optimization problem~(\ref{nfpe-opt-sfc1}) are necessary. Firstly, we omit the payoff associated with the empty sequence in the objective function, as it does not depend on $\gamma^i(\varepsilon)$. Secondly, we exclude redundant inequalities in the constraints that arise from the same recursive equation for $\eta^i(\varepsilon)$ and $\gamma^i(\varepsilon)$, but a smaller value of $\eta^i(\varepsilon)$ is used at the empty sequence.
	
	In accordance with the Nash equilibrium principle, we define $\gamma^*(\varepsilon)$ as a Nash equilibrium of $\Gamma_s(\varepsilon)$ precisely when $\gamma^{*i}(\varepsilon)$ individually solves the optimization problem~(\ref{nfpe-opt-sfc1}) against $\gamma^{*-i}(\varepsilon)$ for every player $i\in N$. By applying the optimality conditions to the problem~(\ref{nfpe-opt-sfc1}) for all players and setting $\hat\gamma(\varepsilon)=\gamma(\varepsilon)$, we derive the polynomial equilibrium system,
	\begin{equation}\label{nfpe-eqt-sfc1}\begin{array}{l}
			g^i(\varpi^i_{I^j_i}a,\gamma^{-i}(\varepsilon))+\lambda^i(\varpi^i_{I^j_i}a)-\nu^i_{I^j_i}=0,\;i\in N,(j,a)\in D_i,\\
			
			g^i(\varpi^i_{I^j_i}a,\gamma^{-i}(\varepsilon))-\nu^i_{I^j_i}+\zeta^i_{I^j_i}(a)=0,\;i\in N,(j,a)\notin D_i,\\
			
			\sum\limits_{a\in A(I^j_i)}\gamma^i(\varepsilon;\varpi^i_{I^j_i}a)-\gamma^i(\varepsilon;\varpi^i_{I^j_i})=0,\;i\in N,j\in M_i,\\
			
			(\gamma^i(\varepsilon;\varpi^i_{I^j_i}a)-\eta^i(\varepsilon;\varpi^i_{I^j_i}a))\lambda^i(\varpi^i_{I^j_i}a)=0,\\
			
			\eta^i(\varepsilon;\varpi^i_{I^j_i}a)\le\gamma^i(\varepsilon;\varpi^i_{I^j_i}a),\;0\le\lambda^i(\varpi^i_{I^j_i}a),\;i\in N,(j,a)\in D_i,
		\end{array}
	\end{equation}
	where $\zeta^i_{I^j_i}(a)=\sum_{j_q\in M_i(\varpi^i_{I^j_i}a)}\nu^i_{I^{j_q}_i}$. Hence, $\gamma^*(\varepsilon)$ is a Nash equilibrium of $\Gamma_s(\varepsilon)$ if and only if there exists a corresponding pair $(\lambda^*,\nu^*)$ such that together with $\gamma^*(\varepsilon)$ they fulfill the system~(\ref{nfpe-eqt-sfc1}). Following Lemma~\ref{nfpe-lem-sfc3}, we derive a specific condition that $\gamma^*(\varepsilon)$ must fulfill, as stated in Lemma~\ref{nfpe-lem-sfc4}.
	\begin{lemma}\label{nfpe-lem-sfc4}{\em
			The profile $\gamma^*(\varepsilon)\in\Lambda(\varepsilon)$ is a Nash equilibrium of $\Gamma_s(\varepsilon)$ if and only if, for each player $i\in N$ and $j\in M_i,a\in A(I^j_i)$, it holds that $\gamma^{*i}(\varepsilon;\varpi^i_{I^j_i}a)=\eta^i(\varepsilon;\varpi^i_{I^j_i}a)$ whenever $\varpi^i_{I^j_i}a$ fails to be a best-response sequence to $\gamma^*(\varepsilon)$.}
	\end{lemma}
	\begin{proof}
		For a given $\gamma^i\in \Lambda^i$ of player $i\in N$, let $y(\gamma^i,\varepsilon)=(y(\gamma^i,\varepsilon;\varpi^i):\varpi^i\in W^i)$, where $y(\gamma^i,\varepsilon;\varpi^i) = \eta^i(\varepsilon;\varpi^i)+(1-\eta^i(\emptyset))\gamma^i(\varpi^i)$. It follows that $y(\cdot,\varepsilon)$ is a bijection from $\Lambda^i$ to $\Lambda^i(\varepsilon)$. We establish in the following an optimization problem~(\ref{nfpe-opt-sfc2}),
		\begin{equation}
			\label{nfpe-opt-sfc2}
			\begin{array}{rl}
				\max\limits_{\gamma^i} &\sum\limits_{j\in M_i}\sum\limits_{a\in A(I^j_i)}\gamma^i(\varpi^i_{I^j_i}a)g^i(\varpi^i_{I^j_i}a,\hat\gamma^{-i}(\varepsilon))\\
				\text{s.t.}&\sum\limits_{a\in A(I^j_i)}\gamma^i(\varpi^i_{I^j_i}a)-\gamma^i(\varpi^i_{I^j_i})=0,\;j\in M_i,\\
				&0\le \gamma^i(\varpi^i_{I^j_i}a),\;(j,a)\in D_i.
			\end{array}
		\end{equation}
		A perturbed realization plan $\gamma^{i}(\varepsilon)$ is a optimal solution to the problem~(\ref{nfpe-opt-sfc1}) if and only if there exists a realization plan $\gamma^{i}$ that optimally solves the problem~(\ref{nfpe-opt-sfc2}) and satisfies $y(\gamma^{i},\varepsilon)=\gamma^{i}(\varepsilon)$. 
		
		For player $i\in N$, there exists some $\gamma^{*i}\in\Lambda^i$ such that $y(\gamma^{*i},\varepsilon)=\gamma^{*i}(\varepsilon)$ and solves against $\gamma^{*-i}(\varepsilon)$ the optimization problem~(\ref{nfpe-opt-sfc2}).  By combining the above discussion with Lemma~\ref{nfpe-lem-sfc3}, we conclude that for any $j\in M_i,a\in A(I^j_i)$ where $\varpi^i_{I^j_i}a$ is not a best-response sequence to $\gamma^*(\varepsilon)$, it holds that $\gamma^{*i}(\varpi^i_{I^j_i}a)=0$. Therefore, $\gamma^{*i}(\varepsilon;\varpi^i_{I^j_i}a)=y(\gamma^{*i},\varepsilon;\varpi^i_{I^j_i}a)=\eta^i(\varepsilon;\varpi^i_{I^j_i}a)$. This completes the proof.
	\end{proof}
	\begin{theorem}\label{nfpe-the-sfc1} {\em A mixed strategy $\sigma^*$ in the pair $(\sigma^*,\gamma^*)\in T$ is a normal-form perfect equilibrium of $\Gamma$ if and only if there exists a sequence of perturbed games in sequence form, $\{\Gamma_s(\varepsilon^k)\}_{k=1}^{\infty}$, with $\lim_{k\to\infty}\varepsilon^k=0$, and a sequence of realization plans $\{\gamma^*(\varepsilon^k)\}_{k=1}^{\infty}$ with each $\gamma^*(\varepsilon^k)$ representing a Nash equilibrium of $\Gamma_s(\varepsilon^k)$, such that $\lim_{k\to\infty}\gamma^*(\varepsilon^k)=\gamma^*$.}
	\end{theorem}
	\begin{proof}
		Assume there exists a sequence of perturbed games  $\{\Gamma_s(\varepsilon^k)\}_{k=1}^{\infty}$ with $\gamma^*(\varepsilon^k)$ being a Nash equilibrium for each $\Gamma_s(\varepsilon^k)$ and $\lim_{k\to\infty}\gamma^*(\varepsilon^k)=\gamma^*$. Let $\{\sigma^k\}_{k=1}^{\infty}$ be a sequence of totally mixed strategies with $(\sigma^k,\gamma^*(\varepsilon^k))\in T$ and $\lim_{k\to\infty}\sigma^k=\sigma^*$. Note that $\sigma^k$ does not necessarily equal $\sigma(\gamma^*(\varepsilon^k))$. For any pure strategy $s^i$ of player $i\in N$, if $u^i(s^i, \sigma^{k-i}) < u^i(\tilde s^{i},\sigma^{k-i})$ holds for some $\tilde s^i\in S^i$, then by Lemma~\ref{nfpe-lem-sfc3}, there exists a sequence $\varpi^i_{I^j_i}a$ for $j\in M_i,a\in A(I^j_i)$ such that $s^i(\varpi^i_{I^j_i}a)=1$ and $\varpi^i_{I^j_i}a$ is not a best-response sequence to $\gamma^*(\varepsilon^k)$. Accordingly, Lemma~\ref{nfpe-lem-sfc4} implies that $\gamma^{*i}(\varepsilon^k;\varpi^i_{I^j_i}a)=\sum_{s^i\in S^i}s^i(\varpi^i_{I^j_i}a)\sigma^{ki}(s^i)=\eta^i(\varepsilon^k;\varpi^i_{I^j_i}a)$, which leads to $\sigma^{ki}(s^i)\leq\eta^i(\varepsilon^k;\varpi^i_{I^j_i}a)\leq\varepsilon^k$. Therefore the sufficiency follows immediately from Definition~\ref{nfpe-def-pre1}.
		
		Conversely, assume that $\sigma^*$ is a normal-form perfect equilibrium of $\Gamma$ and $\gamma^*=\gamma(\sigma^*)$. Then, there exists a sequence of totally mixed strategies $\{\sigma(\varepsilon^k)\}_{k=1}^{\infty}$ such that $\lim_{k\to\infty}\varepsilon^k=0$ and $\lim_{k\to\infty}\sigma(\varepsilon^k)=\sigma^*$, where each $\sigma(\varepsilon^k)$ is an $\varepsilon^k$-normal form perfect equilibrium. Consider a specific $\sigma(\varepsilon^k)$ with sufficiently large $k$, if $u^i(s^i, \sigma^{-i}(\varepsilon^k)) < u^i(\tilde s^{i},\sigma^{-i}(\varepsilon^k))$ holds for some $\tilde s^i\in S^i$, then $\sigma^i(\varepsilon^k;s^i)\leq\varepsilon^k$. Let $\gamma^k=\gamma(\sigma(\varepsilon^k))$ and $\tilde\varepsilon^k=\max_{i\in N}|S^i|\varepsilon^k$, we construct a perturbed game $\Gamma_s(\tilde\varepsilon^k)$ with $\gamma^k$ being a Nash equilibrium. The perturbation $\eta(\tilde\varepsilon^k)=(\eta^i(\tilde\varepsilon^k;\varpi^i):i\in N,\varpi^i\in W^i)$ adheres to the system~(\ref{nfpe-equ-cha1}), defined for $(j,a)\in D_i$ as follows,
		\[\eta^i(\tilde\varepsilon^k;\varpi^i_{I^j_i}a)=\left\{\begin{array}{ll}
			\gamma^{ki}(\varpi^i_{I^j_i}a)& \text{if $\varpi^i_{I^j_i}a$ is not a best-response sequence to $\gamma^k$},\\
			\tilde\varepsilon^k & \text{otherwise.}
		\end{array}\right.\]
		It can be observed that $0\leq\eta^i(\tilde\varepsilon^k;\varpi^i_{I^j_i}a)\leq \tilde\varepsilon^k$. Furthermore, the recursive expressions of $\eta(\tilde\varepsilon^k)$ and $\gamma^k$ ensure that $\eta^i(\tilde\varepsilon^k;\varpi^i_{I^j_i}a)=\gamma^{ki}(\varpi^i_{I^j_i}a)$ holds for all $i\in N,j\in M_i$, and $a\in A(I^j_i)$, provided that $\varpi^i_{I^j_i}a$ is not a best-response sequence to $\gamma^k$. Applying Lemma~\ref{nfpe-lem-sfc4}, we conclude that $\gamma^k$ is a Nash equilibrium for $\Gamma_s(\tilde\varepsilon^k)$. As $\lim_{k\to\infty}\gamma^k=\gamma^*$, the proof is complete.
	\end{proof}

	\section{A logarithmic-barrier smooth path}\label{nfpe-sec-prm3}
	Drawing on the established characterization of normal-form perfect equilibrium in sequence form, this section introduces a differentiable path-following method for computing normal-form perfect equilibria, accompanied by a rigorous theoretical analysis.
	\subsection{An Artificial Game and Equilibrium Convergence Analysis}\label{nfpe-sct-log1}
	Let $\varepsilon=1/\max_{i\in N,j\in M_i}|A(I^j_i)|$ and $\eta^0=(\eta^{0i}(\varpi^i):i\in N,\varpi^i\in W^i)$ be a given vector that satisfies the system~(\ref{nfpe-equ-cha1}). Furthermore, let $\gamma^0=(\gamma^{0i}(\varpi^i):i\in N,\varpi^i\in W^i)$ denote a given realization plan profile with $\gamma^{0i}(\varpi^i)\geq\eta^{0i}(\varpi^i)$, which serves as a starting point for the smooth path discussed later. The proposed method also requires the following functions, defined over the interval $[0,2]$,	\[\rho(t)=\left\{\begin{array}{ll}
		\frac{4}{3}t& t\leq \frac{1}{2},\\
		-\frac{4}{3}(1-t)^2+1& t\leq 1,\\
		1&\text{Otherwise,}
	\end{array}\right.\]
	and 
	\[\theta(t)=\left\{\begin{array}{ll}
		0& t\leq 1,\\
		\frac{4}{3}\left(t-1\right)^2& t\leq \frac{3}{2},\\
		\frac{4}{3}t-\frac{5}{3}&\text{Otherwise,}
	\end{array}\right.\]
 along with $c(t) =\exp(1-1/\rho(t))$.
 
    For $t\in(0,2]$, we constitute a logarithmic-barrier artificial game $\Gamma_{s}^l(t)$ in sequence form where the strategy $\gamma^i(t)$ for each player $i$ is defined by
	\begin{equation}
		\label{logexsfnecst}
		\begin{array}{l}
			\sum\limits_{a\in A(I^j_i)}\gamma^i(t;\varpi^i_{I^j_i}a)-(1- \theta(t))\gamma^i(t;\varpi^i_{I^j_i})-\theta(t)\gamma^{0i}(\varpi^i_{I^j_i})=0,\;j\in M_i,\\
			0\leq\gamma^i(t;\varpi^i_{I^j_i}a),\;j\in M_i,a\in A(I^j_i),
		\end{array}
	\end{equation}
	and $\gamma^i(t;\emptyset)=1$. Let $\Omega(t)=\{(\gamma^i(t):i\in N)|\gamma^i(t)\text{ satisfies the system (\ref{logexsfnecst})}\}$ and $\Omega=\{(\gamma(t),t)|\gamma(t)\in\Omega(t),t\in(0,2]\}$. When $t=2$, $\gamma^i(t;\varpi^i_{I^j_i}a)$ no longer depends on $\gamma^i(t;\varpi^i_{I^j_i})$ for any $i\in N, j\in M_i,a\in A(I^j_i)$. The strategy space $\Omega(t)$ corresponds to the realization plan space for $t\in (0,1]$. In the artificial game $\Gamma_{s}^l(t)$, each player $i$ determines an optimal response to a prescribed strategy $\hat\gamma(t)\in\Omega(t)$ by solving the strictly convex optimization problem,
	\begin{equation}
		\label{nfpe-log-opt-1}
		\begin{array}{rl}
			\max\limits_{\gamma^i(t)} & (1-c(t))\sum\limits_{j\in M_i}\sum\limits_{a\in A(I^j_i)}\gamma^i(t;\varpi^i_{I^j_i}a)g^i(\varpi^i_{I^j_i}a,\hat\gamma^{-i}(t))\\
			&+c(t)\sum\limits_{(j,a)\in D_i}\gamma^{0i}(\varpi^i_{I^j_i}a)\ln(\gamma^i(t;\varpi^i_{I^j_i}a)-\rho(t)(1-\theta(t))\eta^{0i}(\varpi^i_{I^j_i}a))\\
			&+\theta(t)\sum\limits_{(j,a)\notin D_i}\gamma^{0i}(\varpi^i_{I^j_i}a)\ln(\gamma^i(t;\varpi^i_{I^j_i}a))\\
			\text{s.t.} & \sum\limits_{a\in A(I^j_i)}\gamma^i(t;\varpi^i_{I^j_i}a)-(1 - \theta(t))\gamma^i(t;\varpi^i_{I^j_i})-\theta(t)\gamma^{0i}(\varpi^i_{I^j_i})=0,\;j\in M_i.
		\end{array}
	\end{equation}
	Through the application of the optimality conditions to the problem~(\ref{nfpe-log-opt-1}) and the assumption $\hat{\gamma}(t) = \gamma(t)$, we derive the polynomial equilibrium system of $\Gamma_{s}^l(t)$,
	\begin{equation}\label{nfpe-eqt-log1}\begin{array}{l}
			(1-c(t))g^i(\varpi^i_{I^j_i}a,\gamma^{-i}(t))+ \lambda^i(\varpi^i_{I^j_i}a)-\nu^i_{I^j_i} \\
			\hspace{3.0cm}+ (1-\theta(t))\zeta^i_{I^j_i}(a) = 0,\;i\in N,j\in M_i,a\in A(I^j_i),\\
			\sum\limits_{a\in A(I^j_i)}\gamma^i(t;\varpi^i_{I^j_i}a)-(1-\theta(t))\gamma^i(t;\varpi^i_{I^j_i})-\theta(t)\gamma^{0i}(\varpi^i_{I^j_i})=0,\;i\in N,j\in M_i,\\
			(\gamma^i(t;\varpi^i_{I^j_i}a)-\rho(t)(1-\theta(t))\eta^{0i}(\varpi^i_{I^j_i}a))\lambda^i(\varpi^i_{I^j_i}a)=c(t)\gamma^{0i}(\varpi^i_{I^j_i}a),\\
			\hspace{1.0cm}\rho(t)(1-\theta(t))\eta^{0i}(\varpi^i_{I^j_i}a)<\gamma^i(t;\varpi^i_{I^j_i}a),\;0<\lambda^i(\varpi^i_{I^j_i}a),\;i\in N,(j,a)\in D_i,\\
			\gamma^i(t;\varpi^i_{I^j_i}a) \lambda^i(\varpi^i_{I^j_i}a)=\theta(t)\gamma^{0i}(\varpi^i_{I^j_i}a),\\
			\hspace{3.2cm}0<\gamma^i(t;\varpi^i_{I^j_i}a),\;0\le\lambda^i(\varpi^i_{I^j_i}a),\;i\in N,(j,a)\notin D_i,
		\end{array}
	\end{equation}
	where $\zeta^i_{I^j_i}(a)=\sum_{{j_q}\in M_i(\varpi^i_{I^j_i}a)}\nu^i_{I^{j_q}_i}$. One can note that $\gamma^*(t)$ constitutes a solution to the optimization problem~(\ref{nfpe-log-opt-1}) against itself precisely when there exists a unique pair $(\lambda^*, \nu^*)$, along with $\gamma^*(t)$, satisfies the system~(\ref{nfpe-eqt-log1}). For values of $t\in(0,1]$, the condition $\gamma^*(t) \in \text{int}(\Lambda)$ ensures that $\sigma(\gamma^*(t))$ is well-defined and $(\sigma(\gamma^*(t)),\gamma^*(t))\in T$.
	
	Let $\widetilde{\mathscr{C}}_L=\{(\gamma(t),t,\lambda,\nu)|(\gamma(t),t,\lambda,\nu) \text{ satisfies the system~(\ref{nfpe-eqt-log1}) with } 0<t\leq 2\}$ and $\mathscr{C}_L$ be the closure of $\widetilde{\mathscr{C}}_L$. To analyze the equilibrium convergence, it is adequate to consider the phase $t\in(0,1]$ where $\theta(t)=0$ and the strategies are realization plans. We denote the subset of $\mathscr{C}_L$ corresponding to $t\in[0,1]$ as $\mathscr{C}^{R}_L$. Subsequently, we introduce a theorem from Luo and Luo~\cite{luoExtensionHoffmansError1994} that is essential for our analysis.
	
	\begin{theorem}\label{nfpe-thm-log1}
		{\em Let $V$ denote the set of $v\in\mathbb{R}^{n_0}$ satisfying $f_1(v) \leq 0,\cdots,f_l(v)\leq 0, p_1(v) = 0,\cdots,p_q(v)=0$, where each $f_i$ and $p_j$ is a polynomial with real coefficients. Suppose that $V$ is nonempty. Then there exist constants $\tau > 0,\kappa > 0$ and $\kappa' > 0$ such that $dist(v,V) \leq \tau(1+\|v\|)^{\kappa'}(\|[f (v)]_+\|+\|p(v)\|)^\kappa$ for any $v\in\mathbb{R}^{n_0}$. Here $dist(\cdot,\cdot)$ denotes the Euclidean distance function between two sets, $f(v) = (f_1(v),\cdots, f_l(v))^\top$, $p(v) = (p_1(v),\cdots, p_q(v))^\top$, and $[\cdot]_+$ denotes the positive part of a vector.}
	\end{theorem}
	Consider the equilibrium system~(\ref{nfpe-eqt-sfc1}) of $\Gamma_s(\varepsilon)$. For $t\in [0,1]$, by setting $\varepsilon=t$ and $\eta(t)=\rho(t)\eta^0$, we obtain a particular perturbed game $\Gamma_s(t)$ and its corresponding equilibrium system. Let $\mathscr{C}_E$ represent the set of tuples $(\gamma(t),t,\lambda,\nu)$ that solve the equilibrium system of $\Gamma_s(t)$. As a direct application of Theorem~\ref{nfpe-thm-log1}, we can infer the following conclusions.
	\begin{corollary}\label{nfpe-cor-log1}
		{\em For $(\gamma,t,\lambda,\mu)\in\mathbb{R}^{n_0}\times[0,1]\times\mathbb{R}^{n_0}\times\mathbb{R}^{m_0}$, let $f_1(\gamma,t,\lambda,\nu)=(f_1(\gamma,t,\lambda,\nu;\varpi^i_{I^j_i}a):i\in N,(j,a)\in D_i)$ with $f_1(\gamma,t, \lambda,\nu;\varpi^i_{I^j_i}a)=\rho(t)\eta^{0i}(\varpi^i_{I^j_i}a)-\gamma^i(\varpi^i_{I^j_i}a)$, $f_2(\gamma,t,\lambda,\nu)=(f_2(\gamma,t,\lambda,\nu;\varpi^i_{I^j_i}a):i\in N,(j,a)\in D_i)$ with $f_2(\gamma,t,\lambda,\nu;\varpi^i_{I^j_i}a)=- \lambda^i(\varpi^i_{I^j_i}a)$. Furthermore, define $p_1(\gamma,t,\lambda,\nu)=(p_1(\gamma,t,\lambda,\nu;\varpi^i_{I^j_i}a):i\in N,(j,a)\in D_i)$ with $p_1(\gamma,t,\lambda,\nu;\varpi^i_{I^j_i}a)=g^i(\varpi^i_{I^j_i}a,\gamma^{-i})+ \lambda^i(\varpi^i_{I^j_i}a)
		-\nu^i_{I^j_i}$, $p_2(\gamma,t,\lambda,\nu)=(p_2(\gamma,t,\lambda,\nu;\varpi^i_{I^j_i}a):i\in N,(j,a)\notin D_i)$ with $p_2(\gamma,t,\lambda,\nu;\varpi^i_{I^j_i}a)=g^i(\varpi^i_{I^j_i}a,\gamma^{-i})-\nu^i_{I^j_i}+\zeta^i_{I^j_i}(a)$, $p_3(\gamma,t,\lambda,\nu)=(p_3(\gamma,t,\lambda,\nu;\varpi^i_{I^j_i}a):i\in N,(j,a)\in D_i)$ with $p_3(\gamma,t,\lambda,\nu;\varpi^i_{I^j_i}a)=(\gamma^i(\varpi^i_{I^j_i}a)-\rho(t)\eta^{0i}(\varpi^i_{I^j_i}a))\lambda^i(\varpi^i_{I^j_i}a)$, $p_4(\gamma,t,\lambda,\nu)=(p_4(\gamma,t,\lambda,\nu;\varpi^i_{I^j_i}):i\in N,j\in M_i)$ with $p_4(\gamma,t,\lambda,\nu;\varpi^i_{I^j_i})=\sum_{a\in A(I^j_i)}\gamma^i(\varpi^i_{I^j_i}a)-\gamma^i(\varpi^i_{I^j_i})$. Then there exist constants $\tau_1 > 0,\kappa_1 > 0$ and $\kappa'_1 > 0$ such that 
		\[\begin{array}{l}dist((\gamma,t, \lambda,\nu),\mathscr{C}_E)=\operatorname*{min}_{(\tilde{\gamma}(t),\tilde{t},\tilde{\lambda},\tilde{\nu})\in\mathscr{C}_E}\|(\gamma,t, \lambda,\nu)-(\tilde{\gamma}(t),\tilde{t},\tilde{\lambda},\tilde{\nu})\|\\
		\hspace{1.5cm}\leq \tau_{1}(1+\|(\gamma,t,\lambda,\nu)\|)^{\kappa_{1}^{\prime}}(\|[f_1(\gamma,t, \lambda,\nu)]_{+}\|+\|[f_2(\gamma,t,\lambda,\nu)]_{+}\|\\
		\hspace{2cm}+\|p_1(\gamma,t,\lambda,\nu)\|+\|p_2(\gamma,t,\lambda,\nu)\|+\|p_3(\gamma,t,\lambda,\nu)\|+\|p_4(\gamma,t,\lambda,\nu)\|)^{\kappa_{1}}\end{array}\]
		for any $(\gamma,t,\lambda,\mu)\in\mathbb{R}^{n_0}\times[0,1]\times\mathbb{R}^{n_0}\times\mathbb{R}^{m_0}$.}
	\end{corollary}
	\begin{lemma}\label{nfpe-lem-log3}
		{\em There exist constants $\tau_L>0$ and $\kappa_L> 0$ such that
		\[\min_{(\tilde{\gamma}(t),\tilde{t},\tilde{\lambda},\tilde{\nu})\in\mathscr{C}_E}\|(\gamma(t),t, \lambda,\nu)-(\tilde{\gamma}(t),\tilde{t},\tilde{\lambda},\tilde{\nu})\|\leq\tau_Lc(t)^{\kappa_L}\]
		for any $(\gamma(t),t,\lambda,\nu)\in\mathscr{C}^R_L$.
		}
	\end{lemma}
	\begin{proof}
		When $t\in (0,1]$, the system (\ref{nfpe-eqt-log1}) can be rewritten as
		\begin{equation}\label{nfpe-eqt-log2}\begin{array}{l}
				g^i(\varpi^i_{I^j_i}a,\gamma^{-i}(t))+ \lambda^i(\varpi^i_{I^j_i}a)-\nu^i_{I^j_i} = c(t)g^i(\varpi^i_{I^j_i}a,\gamma^{-i}(t)),\;i\in N,(j,a)\in D_i,\\
				g^i(\varpi^i_{I^j_i}a,\gamma^{-i}(t))-\nu^i_{I^j_i} + \zeta^i_{I^j_i}(a) = c(t)g^i(\varpi^i_{I^j_i}a,\gamma^{-i}(t)),\;i\in N,(j,a)\notin D_i,\\
				\sum\limits_{a\in A(I^j_i)}\gamma^i(t;\varpi^i_{I^j_i}a)-\gamma^i(t;\varpi^i_{I^j_i})=0,\;i\in N,j\in M_i,\\
				(\gamma^i(t;\varpi^i_{I^j_i}a)-\rho(t)\eta^{0i}(\varpi^i_{I^j_i}a)) \lambda^i(\varpi^i_{I^j_i}a)=c(t)\gamma^{0i}(\varpi^i_{I^j_i}a),\\
				\hspace{1.7cm}\rho(t)\eta^{0i}(\varpi^i_{I^j_i}a)<\gamma^i(t;\varpi^i_{I^j_i}a),\;0<\lambda^i(\varpi^i_{I^j_i}a),\;i\in N,(j,a)\in D_i.
			\end{array}
		\end{equation}
		Applying Corollary~\ref{nfpe-cor-log1} to the system~(\ref{nfpe-eqt-log2}) reveals that there exist constants $\tau_1> 0, \kappa_1 > 0$, and $\kappa^{\prime}_1 > 0$ such that, for any $(\gamma(t),t,\lambda,\nu) \in\mathscr{C}^R_L$, 
		\[\begin{array}{l}	\operatorname*{min}\limits_{(\tilde{\gamma}(t),\tilde{t},\tilde{\lambda},\tilde{\nu})\in\mathscr{C}_E}\|(\gamma(t),t,\lambda,\nu)-(\tilde{\gamma}(t),\tilde{t},\tilde{\lambda},\tilde{\nu})\|\leq \tau_1(1+\|(\gamma(t),t,\lambda,\nu)\|)^{\kappa_{1}^{\prime}}c(t)^{\kappa_{1}}\\
		\hspace{2.5cm}(\sum\limits_{i\in N}\sum\limits_{j\in M_{i},a\in A(I^j_i)}\|g^i(\varpi^i_{I^j_i}a,\gamma^{-i}(t))\|+\sum\limits_{i\in N}\sum\limits_{(j,a)\in D_i}\gamma^{0i}(\varpi^i_{I^j_i}a))^{\kappa_1}.
		\end{array}\]
		Let $\kappa_L = \kappa_1$ and
		\[\begin{array}{l}\tau_{L}=\tau_{1}\max\limits_{(\gamma(t),t,\lambda,\nu)\in\mathscr{C}^R_{L}}(1+\|(\gamma(t),t,\lambda,\nu)\|)^{\kappa_{1}^{\prime}}\\
			\hspace{2.5cm}(\sum\limits_{i\in N}\sum\limits_{j\in M_{i},a\in A(I^j_i)}\|g^i(\varpi^i_{I^j_i}a,\gamma^{-i}(t))\|+\sum\limits_{i\in N}\sum\limits_{(j,a)\in D_i}\gamma^{0i}(\varpi^i_{I^j_i}a)))^{\kappa_1}.\end{array}\]
		The compactness of $\mathscr{C}^R_L$ demonstrated in Appendix~\ref{nfpe-appen-2} confirms that $\tau_{L}$ is finite. Thus \[\min_{(\tilde{\gamma}(t),\tilde{t},\tilde{\lambda},\tilde{\nu})\in\mathscr{C}_E}\|(\gamma(t),t,\lambda,\nu)-(\tilde{\gamma}(t),\tilde{t},\tilde{\lambda},\tilde{\nu})\|\leq\tau_Lc(t)^{\kappa_L}\]for any $(\gamma(t),t,\lambda,\nu)\in\mathscr{C}^R_L$. The proof is completed.
	\end{proof}
	\begin{theorem}\label{nfpe-thm-log5}{\em
			Let $\{\gamma^*(t_k)\}_{k=1}^{\infty}$ be a sequence of Nash equilibria of $\Gamma_{s}^l(t)$ defined by the system~(\ref{nfpe-eqt-log1}) with $t=t_k\in(0,1]$ and $\lim_{k\to \infty}t_k=0$. Then every limit point of the totally mixed strategy sequence $\{\sigma(\gamma^*(t_k))\}_{k=1}^{\infty}$ yields a normal-form perfect equilibrium.
		}
	\end{theorem}
	\begin{proof}
		For each $\gamma^*(t_k)$, let $(\lambda^*(t_k),\nu^*(t_k))$ be the associated pair satisfying the system~(\ref{nfpe-eqt-log1}). Due to the compactness of $\mathscr{C}^R_L$, the sequence $\{(\gamma^*(t_k),t_k,\lambda^*(t_k),\nu^*(t_k))\}_{k=1}^{\infty}$ has a convergent subsequence, which we denote in the same notation. As to each $t_k$, assume that \[(\tilde{\gamma}(\tilde{t}_k),\tilde{t}_k,\tilde{ \lambda}(\tilde{t}_k),\tilde{\nu}(\tilde{t}_k))\in\underset{(\tilde{\gamma}(\tilde t),\tilde{t},\tilde{\lambda},\tilde{\nu})\in\mathscr{C}_E}{\arg\,\min}\;\|(\gamma^*(t_k),t_k,\lambda^*(t_k),\nu^*(t_k))-(\tilde{\gamma}(\tilde{t}),\tilde{t},\tilde{\lambda},\tilde{\nu})\|.\]Lemma~\ref{nfpe-lem-log3} implies that there exist constants $\tau_L>0$ and $\kappa_L> 0$ such that $|\tilde{t}_k-t_k|\leq\tau_Lc(t_k)^{\kappa_L}$. Given that $c(t_k)=\exp(1-1/\rho(t_k))$, a sufficiently large constant $K_0$ ensures  $\tau_Lc(t_k)^{\kappa_L}<t_k$ when $k>K_0$. It follows that $\tilde{t}_k>0$ for all $k>K_0$. Consequently, we deduce that $\lim_{k\to\infty} \sigma(\gamma^*(t_k))=\lim_{k\to\infty}\sigma(\tilde\gamma(\tilde t_k))$, indicating a normal-form perfect equilibrium. This completes the proof.
	\end{proof}
	\subsection{A Smooth Path to a Normal-Form Perfect Equilibrium}\label{SPne}
	In this subsection, we establish the existence of a smooth path consisting of solutions to the system~(\ref{nfpe-eqt-log1}). This path starts from an arbitrary initial realization plan and ultimately converge into a normal-form perfect equilibrium.
	
	\begin{lemma}\label{nfpe-lem-log4} {\em At $t=2$, the system~(\ref{nfpe-eqt-log1}) has a unique solution, given by $(\gamma^*(2),\lambda^*(2),\nu^*(2))$, with the components satisfying $\gamma^{*i}(2;\varpi^i_{I^j_i}a)=\gamma^{0i}(\varpi^i_{I^j_i}a)$, $\lambda^i(2;\varpi^{*i}_{I^j_i}a)=1$ and $\nu^{*i}_{I^j_i}(2)=1$.}
	\end{lemma}
	\begin{proof}
		At $t=2$, the system~(\ref{nfpe-eqt-log1}) can be expressed as follows,
		\begin{equation}\label{nfpe-log-equ-2}\begin{array}{l}
				 \lambda^i(\varpi^i_{I^j_i}a)-\nu^i_{I^j_i}=0,\; i\in N,j\in M_i,a\in A(I^j_i),\\
				\sum\limits_{a\in A(I^j_i)}\gamma^i(2;\varpi^i_{I^j_i}a)-\gamma^{0i}(\varpi^i_{I^j_i})=0,\;i\in N,j\in M_i,\\
				\gamma^i(2;\varpi^i_{I^j_i}a) \lambda^i(\varpi^i_{I^j_i}a)=\gamma^{0i}(\varpi^i_{I^j_i}a),\;i\in N,j\in M_i,a\in A(I^j_i),\\
				0<\gamma^i(2;\varpi^i_{I^j_i}a),\;0<\lambda^i(\varpi^i_{I^j_i}a),\;i\in N,j\in M_i,a\in A(I^j_i).
			\end{array}
		\end{equation}
		Suppose that $(\gamma^*(2),\lambda^*(2),\nu^*(2))$ is a solution to the system~(\ref{nfpe-log-equ-2}). We proceed by substituting the first set of equations into the third set, and summing over $a\in A(I^j_i)$ for each $i\in N,j\in M_i$, obtaining $\nu^{*i}_{I^j_i}(2) = 1$. Consequently, the first set of equations yields $ \lambda^{*i}(2;\varpi^i_{I^j_i}a) = \nu^{*i}_{I^j_i}(2) = 1$ for $i\in N,j\in M_i,a\in A(I^j_i)$. Substituting these results back into the third set, it follows that $\gamma^{*i}(2;\varpi^i_{I^j_i}a) = \gamma^{0i}(\varpi^i_{I^j_i}a)$ for $i\in N,j\in M_i,a\in A(I^j_i)$. The proof is complete.
	\end{proof}
	
	As shown in Lemma~\ref{nfpe-lem-log4}, the system~(\ref{nfpe-eqt-log1}) admits a unique solution at $t=2$. We then identify a connected component containing this solution that extends to intersect the $t=0$ level. To facilitate this, it is crucial to introduce Browder's fixed-point theorem~\cite{Browdercontinuityfixedpoints1960}.
	\begin{theorem}\label{nfpe-thm-log2}{\em Let $V$ be a nonempty, compact and convex subset of $\mathbb{R}^m$ and $f:V\times[0,1]\to V$ be a continuous function. Then the set $F =\{(v,t)\in V\times[0,1]|v=f(v,t)\}$ contains a connected set $F^c$ such that $V\times\{1\}\cap F^c\neq\emptyset$ and $V\times\{0\}\cap F^c\neq\emptyset$.}
	\end{theorem}
	Leveraging the results of Theorem~\ref{nfpe-thm-log2}, we arrive at Theorem~\ref{nfpe-thm-log3}.
	\begin{theorem}\label{nfpe-thm-log3}{\em There is a connected component in $\mathscr{C}_L$ intersecting both $\mathbb{R}^{n_0}\times\{2\}\times\mathbb{R}^{n_0}\times\mathbb{R}^{m_0}$ and $\mathbb{R}^{n_0}\times\{0\}\times\mathbb{R}^{n_0}\times\mathbb{R}^{m_0}$.}
	\end{theorem}
	\begin{proof}
		Let $\Omega_L=\cup_{t\in(0,2]}\Omega(t)$. For $(\hat{\gamma},t)\in \Omega_L\times(0,2]$, let $\varphi(\hat{\gamma},t)=(\gamma^i(t):i\in N)$, where $\gamma^i(t)$ represents the unique solution of the strictly convex optimization problem,
		\begin{equation}
			\label{nfpe-log-opt-3}
			\begin{array}{rl}
				\max\limits_{\gamma^i(t)} & (1-c(t))\sum\limits_{j\in M_i}\sum\limits_{a\in A(I^j_i)}\gamma^i(t;\varpi^i_{I^j_i}a)g^i(\varpi^i_{I^j_i}a,\hat\gamma^{-i}))\\
				&+c(t)\sum\limits_{(j,a)\in D_i}\gamma^{0i}(\varpi^i_{I^j_i}a)\ln(\gamma^i(t;\varpi^i_{I^j_i}a)-\rho(t)(1-\theta(t))\eta^{0i}(\varpi^i_{I^j_i}a))\\
				&+\theta(t)\sum\limits_{(j,a)\notin D_i}\gamma^{0i}(\varpi^i_{I^j_i}a)\ln(\gamma^i(t;\varpi^i_{I^j_i}a))-\frac{1}{2}\sum\limits_{j\in M_i}\sum\limits_{a\in A(I^j_i)}(\gamma^i(t;\varpi^i_{I^j_i}a)-\hat\gamma^i(\varpi^i_{I^j_i}a))^2\\
				\text{s.t.} & \sum\limits_{a\in A(I^j_i)}\gamma^i(t;\varpi^i_{I^j_i}a)-(1 - \theta(t))\gamma^i(t;\varpi^i_{I^j_i})-\theta(t)\gamma^{0i}(\varpi^i_{I^j_i})=0,\;j\in M_i.
			\end{array}
		\end{equation}
		For $(\hat{\gamma},0)\in \Omega_L\times\{0\}$, define $\varphi(\hat{\gamma},0)=(\gamma^i(0):i\in N)$, where $\gamma^i(0)$ corresponds to the unique solution of the strictly convex optimization problem,
		\begin{equation}
			\label{logbsfgop3}
			\begin{array}{rl}
				\max\limits_{\gamma^i(0)} & \sum\limits_{j\in M_i}\sum\limits_{a\in A(I^j_i)}\gamma^i(0;\varpi^i_{I^j_i}a)g^i(\varpi^i_{I^j_i}a,\hat\gamma^{-i})-\frac{1}{2}\sum\limits_{j\in M_i}\sum\limits_{a\in A(I^j_i)}(\gamma^i(0;\varpi^i_{I^j_i}a)-\hat\gamma^i(\varpi^i_{I^j_i}a))^2\\
				\text{s.t.} & \sum\limits_{a\in A(I^j_i)}\gamma^i(0;\varpi^i_{I^j_i}a)-\gamma^i(0;\varpi^i_{I^j_i})=0,\;j\in M_i.
			\end{array}
		\end{equation}According to Theorem 2.2.2 in~\cite{FiaccoIntroductionSensitivityStability1983}, the mapping $\varphi(\gamma, t)$ is continuous from $\Omega_L \times [0,2]$ into $\Omega_L$. Define the set $\mathscr{F} = \{(\gamma(t), t) \in \Omega_L \times [0,2] \mid \varphi(\gamma(t), t) = \gamma(t)\}$. Theorem~\ref{nfpe-thm-log2} guarantees the existence of a connected component of $\mathscr{F}$ that intersects both $\mathbb{R}^{n_0} \times \{2\}$ and $\mathbb{R}^{n_0} \times \{0\}$. We designate this connected component by $\mathscr{F}^c$, and specifically identify the subset corresponding to $t>0$ as $\widetilde{\mathscr{F}}^c$.
		
		By applying the optimality conditions to the problem~(\ref{nfpe-log-opt-3}), we obtain a polynomial system equivalent to the system~(\ref{nfpe-eqt-log1}). Consequently, for each $(\gamma(t), t) \in \widetilde{\mathscr{F}}^c$, there exists a unique pair $(\lambda, \nu)$ such that the system~(\ref{nfpe-eqt-log1}) holds. Define $\widetilde{\mathscr{C}}^c_L=\{(\gamma(t),t,\lambda,\nu)\in \widetilde{\mathscr{C}}_L|(\gamma(t),t)\in \widetilde{\mathscr{F}}^c\}$ and let $\mathscr{C}^c_L$ be its closure. From the previous arguments, it follows that $\mathscr{C}^c_L$ constitutes a connected component within $\mathscr{C}_L$ intersecting $\mathbb{R}^{n_0} \times \{2\} \times \mathbb{R}^{n_0} \times \mathbb{R}^{m_0}$. Let $\{(\gamma(t_k), t_k)\}_{k=1}^\infty \subseteq \widetilde{\mathscr{F}}^c$ be a sequence converging with $\lim_{k\to\infty} t_k = 0$. For each element $(\gamma(t_k), t_k)$, we associate a corresponding pair $(\lambda(t_k), \mu(t_k))$ such that $(\gamma(t_k), t_k, \lambda(t_k), \mu(t_k)) \in \widetilde{\mathscr{C}}^c_L$. Since $\widetilde{\mathscr{C}}^c_L$ is bounded, there exists a convergent subsequence of $\{(\gamma(t_k),t_k,\lambda(t_k),\mu(t_k))\}_{k=1}^{\infty}$. Therefore, the set $\mathscr{C}^c_L$ intersects $\mathbb{R}^{n_0} \times \{0\} \times \mathbb{R}^{n_0} \times \mathbb{R}^{m_0}$, which concludes the proof.
	\end{proof}
	Lemma~\ref{nfpe-lem-log4} asserts that the connected component discussed in Theorem~\ref{nfpe-thm-log3} is uniquely determined and intersects the level $t=2$ at the point $(\gamma^*(2),2,\lambda^*(2),\nu^*(2))$. In order to eliminate the impact of $\gamma^i(t;\varpi^i_{I^j_i}a)\lambda^i(\varpi^i_{I^j_i}a)=0$ for $i\in N,(j,a)\notin D_i$ on the differentiability of the path over the interval $(0,1]$, and reduce the number of variables for more efficient computation, we employ a variable substitution technique as outlined in Cao et al.~\cite{Caodifferentiablepathfollowingmethod2022}. Given $\tau_0>0$ and $\kappa_0>1$, define the following functions,
	\begin{small} \[\psi_1(v,r;\tau_0,\kappa_0)=\left(\frac{v+\sqrt{v^2+4\tau_0r}}{2}\right)^{\kappa_0} \text{ and } \psi_2(v,r;\tau_0,\kappa_0)=\left(\frac{-v+\sqrt{v^2+4\tau_0r}}{2}\right)^{\kappa_0}.\]\end{small}It can be deduced that $\psi_1(v,r;\tau_0,\kappa_0)\psi_2(v,r;\tau_0,\kappa_0)=(\tau_0r)^{\kappa_0}$. Since $\kappa_0>1$, both functions possess continuous differentiability over the domain $\mathbb{R}\times(0,\infty)$. For $x=(x^i(\varpi^i_{I^j_i}a):i\in N,j\in M_i,a\in A(I^j_i))\in \mathbb{R}^{n_0}$, we define $\gamma(x,t)=(\gamma^i(x,t;\varpi^i_{I^j_i}a):i\in N,j\in M_i,a\in A(I^j_i))$ and $\lambda(x,t)=(\lambda^i(x,t;\varpi^i_{I^j_i}a):i\in N,j\in M_i,a\in A(I^j_i))$, where
	\begin{equation}\label{nfpe-eqt-log3}\begin{array}{l}
			\gamma^i(x,t;\varpi^i_{I^j_i}a)=\left\{\begin{array}{ll}
				\rho(t)(1-\theta(t))\eta^{0i}(\varpi^i_{I^j_i}a)&\\\hspace{0.6cm}+\psi_1(x^i(\varpi^i_{I^j_i}a),c(t)^{1/\kappa_0}; \gamma^{0i}(\varpi^i_{I^j_i}a)^{1/\kappa_0}, \kappa_0) & (j,a)\in D_i,\\
				\psi_1(x^i(\varpi^i_{I^j_i}a),\theta(t)^{1/\kappa_0}; \gamma^{0i}(\varpi^i_{I^j_i}a)^{1/\kappa_0},\kappa_0) & (j,a)\notin D_i,
			\end{array}\right.\\
			\lambda^i(x,t;\varpi^i_{I^j_i}a)=\left\{\begin{array}{ll}
				\psi_2(x^i(\varpi^i_{I^j_i}a),c(t)^{1/\kappa_0};\gamma^{0i}(\varpi^i_{I^j_i}a)^{1/\kappa_0},\kappa_0) & (j,a)\in D_i,\\
				\psi_2(x^i(\varpi^i_{I^j_i}a),\theta(t)^{1/\kappa_0};\gamma^{0i}(\varpi^i_{I^j_i}a)^{1/\kappa_0}, \kappa_0) & (j,a)\notin D_i.
			\end{array}\right.
	\end{array}\end{equation}
	It is evident that $(\gamma^i(x,t;\varpi^i_{I^j_i}a)-\rho(t)(1-\theta(t))\eta^{0i}(\varpi^i_{I^j_i}a)) \lambda^i(x,t;\varpi^i_{I^j_i}a)=c(t)\gamma^{0i}(\varpi^i_{I^j_i}a)$ for $i\in N,(j,a)\in D_i$ and $\gamma^i(x,t;\varpi^i_{I^j_i}a) \lambda^i(x,t;\varpi^i_{I^j_i}a)=\theta(t)\gamma^{0i}(\varpi^i_{I^j_i}a)$ for $i\in N,(j,a)\notin D_i$. Let $\alpha=(\alpha(\varpi^i_{I^j_i}a):i\in N, j\in M_i, a\in A(I^j_i))\in\mathbb{R}^{n_0}$ be an arbitrary vector with $||\alpha||$ sufficiently small. By replacing $\gamma^i(t;\varpi^i_{I^j_i}a)$ and $\lambda^i(\varpi^i_{I^j_i}a)$ with $\gamma^i(x,t;\varpi^i_{I^j_i}a)$ and $ \lambda^i(x,t;\varpi^i_{I^j_i}a)$ in the system~(\ref{nfpe-eqt-log1}), and subtracting $c(t)(1-\theta(t))\alpha$ from the left-hand side of the first set of equations, we derive an equivalent system formulated with fewer variables and constraints,
	\begin{equation}\label{nfpe-eqt-log4}\begin{array}{l}
			(1-c(t))g^i(\varpi^i_{I^j_i}a,\gamma^{-i}(x,t))+ \lambda^i(x,t;\varpi^i_{I^j_i}a)-\nu^i_{I^j_i} + (1-\theta(t))\zeta^i_{I^j_i}(a)\\
			\hspace{3cm}-c(t)(1-\theta(t))\alpha(\varpi^i_{I^j_i}a)=0,\;i\in N,j\in M_i,a\in A(I^j_i),\\
			\sum\limits_{a\in A(I^j_i)}\gamma^i(x,t;\varpi^i_{I^j_i}a)-(1 - \theta(t))\gamma^i(x,t;\varpi^i_{I^j_i})-\theta(t)\gamma^{0i}(\varpi^i_{I^j_i})=0, i\in N,j\in M_i.\\
		\end{array}
	\end{equation}
	The parameter $\alpha$ is introduced to address degenerate cases and has no effect on the convergence analysis in Subsection~\ref{nfpe-sct-log1}. At $t=2$, the system~(\ref{nfpe-eqt-log4}) has a unique solution $(x^*(2),2,\nu^*(2))$ with $x^{*i}(2;\varpi^i_{I^j_i}a)=\gamma^{0i}(\varpi^i_{I^j_i}a)^{1/\kappa_0}-1$ and $\nu^{*i}_{I^j_i}(2)=1$. Let $\widetilde{\mathscr{P}}_L=\{(x,t,\nu)|(x,t,\nu)\text{ satisfies the system~(\ref{nfpe-eqt-log4}) with } 0<t\leq 2\}$ and $\mathscr{P}_L$ be the closure of $\widetilde{\mathscr{P}}_L$. We have the following theorem.
	\begin{theorem}\label{nfpe-thm-log4}{\em Given almost any $\alpha\in\mathbb{R}^{n_0}$ with sufficiently small $||\alpha||$, a smooth path can be identified in $\mathscr{P}_L$ that begins at $(x^*(2), 2, \nu^*(2))$ when $t = 2$ and leads to a normal-form perfect equilibrium as $t$ approaches zero.}
	\end{theorem}
	\begin{proof}
		Let $p(x,t,\nu;\alpha)$ denote the left-hand sides of the equations in the system~(\ref{nfpe-eqt-log4}), and define $p_\alpha(x,t,\nu) = p(x,t,\nu;\alpha)$ when $\alpha$ is treated as a constant. The function $p(x,t,\nu;\alpha)$ is continuously differentiable over $\mathbb{R}^{n_0}\times(0,2)\times\mathbb{R}^{m_0}\times\mathbb{R}^{n_0}$. As demonstrated in Appendix~\ref{nfpe-appen-1}, its Jacobian matrix has full-row rank in this region. By using the transversality theorem as stated by Eaves and Schmedders~\cite{EavesGeneralequilibriummodels1999}, one can prove that zero is a regular value of $p_\alpha(x,t,\nu)$ over $\mathbb{R}^{n_0}\times(0,2)\times\mathbb{R}^{m_0}$ for almost any $\alpha$ with $||\alpha||<\epsilon$, where $\epsilon$ is a sufficiently small positive constant.
		
		We select an appropriate parameter $\alpha$ such that zero is a regular value of $p_\alpha(x,t,\nu)$ over $\mathbb{R}^{n_0}\times(0,2)\times\mathbb{R}^{m_0}$. Applying the implicit function theorem, it follows that the component described in Theorem~\ref{nfpe-thm-log3} defines a smooth path in $\mathscr{P}_L$, originating at $((x^*(2), 2,\nu^*(2))$ and terminates at $t=0$. As demonstrated in Appendix~\ref{nfpe-appen-1}, the regularity of zero is preserved at $t=2$ within $\mathbb{R}^{n_0}\times\mathbb{R}^{m_0}$, thereby precluding tangential intersections with $\mathbb{R}^{n_0}\times\{2\}\times\mathbb{R}^{m_0}$. Finally, Theorem~\ref{nfpe-thm-log5} confirms that this smooth path culminates in a normal-form perfect equilibrium. This completes the proof.
	\end{proof}

	\section{Extension of Harsanyi's tracing procedures to the sequence form}\label{nfpe-sec-prm5}
		According to the analysis in \cite{vandenelzenAlgorithmicApproachTracing1999}, the tracing procedure proposed by von Stengel et al. \cite{VonStengelComputingNormalForm2002} for two-player extensive-form games is equivalent to Harsanyi's linear tracing procedure when applied in sequence form. In subsection~\ref{nfpe-sec-htp1}, we extend this tracing procedure to accommodate $n$-player games and develop an alternative differentiable path-following method for computing normal form perfect equilibria. Furthermore, to improve the smoothness of the equilibrium paths, we extend Harsanyi's logarithmic tracing procedure to the sequence form in subsection~\ref{nfpe-sec-htp2}, providing an additional differentiable path-following method.
	\subsection{Harsanyi's linear tracing procedure in sequence form}\label{nfpe-sec-htp1}
	Let $p^0=(p^{0i}(\varpi^i):i\in N,\varpi^i\in W^i)$ be a given prior realization plan profile. For $\gamma\in \Lambda$ and $t\in [0,1]$, define $y(\gamma,t)=(y(\gamma^i,t;\varpi^i):i\in N,\varpi^i\in W^i)$ with $y(\gamma^i,t;\varpi^i)=(1-t)\gamma^i(\varpi^i)+tp^{0i}(\varpi^i)$. This construction ensures $y(\gamma,t)\in \Lambda$. For $t\in(0,1]$, we introduce an artificial game $\Gamma_{s}^v(t)$ in sequence form, where each player $i$ finds their best response to $\hat\gamma\in \Lambda$ by solving the following linear optimization problem,
	\begin{equation}
		\label{nfpe-opt-ev1}
		\begin{array}{rl}
			\max\limits_{\gamma^i} &\sum\limits_{j\in M_i}\sum\limits_{a\in A(I^j_i)}\gamma^i(\varpi^i_{I^j_i}a)g^i(\varpi^i_{I^j_i}a,y(\hat\gamma^{-i},t))\\
			\text{s.t.}&\sum\limits_{a\in A(I^j_i)}\gamma^i(\varpi^i_{I^j_i}a)-\gamma^i(\varpi^i_{I^j_i})=0,\;j\in M_i,\\
			&0\le \gamma^i(\varpi^i_{I^j_i}a),\;(j,a)\in  D_i.
		\end{array}
	\end{equation}
	By applying the optimality conditions to the problem (\ref{nfpe-opt-ev1}) and incorporating the fixed-point argument $\hat\gamma=\gamma$, we derive the polynomial equilibrium system of $\Gamma_{s}^v(t)$,
	\begin{equation}\label{nfpe-eqt-ev1}\begin{array}{l}
			g^i(\varpi^i_{I^j_i}a,y(\gamma^{-i},t))+ \lambda^i(\varpi^i_{I^j_i}a)
			-\nu^i_{I^j_i}=0,\;i\in N,(j,a)\in D_i,\\
			
			g^i(\varpi^i_{I^j_i}a,y(\gamma^{-i},t))-\nu^i_{I^j_i}+\zeta^i_{I^j_i}(a)=0,\;i\in N,(j,a)\notin D_i,\\
			
			\sum\limits_{a\in A(I^j_i)}\gamma^i(\varpi^i_{I^j_i}a)-\gamma^i(\varpi^i_{I^j_i})=0,\;i\in N,j\in M_i,\\
			
			\gamma^i(\varpi^i_{I^j_i}a) \lambda^i(\varpi^i_{I^j_i}a)=0,\;0\le\gamma^i(\varpi^i_{I^j_i}a),\;0\le \lambda^i(\varpi^i_{I^j_i}a),\;i\in N,(j,a)\in D_i,
		\end{array}
	\end{equation}
	where
	$\zeta^i_{I^j_i}(a)=\sum_{j_q\in M_i(\varpi^i_{I^j_i}a)}\nu^i_{I^{j_q}_i}$. As a result, the solution $\gamma^*$ solves the optimization problem~(\ref{nfpe-opt-ev1}) against itself if and only if a pair $(\lambda^*, \nu^*)$ exists together with $\gamma^*$ that collectively satisfies the system~(\ref{nfpe-eqt-ev1}). Referring to the proof of Lemma~\ref{nfpe-lem-sfc4}, we deduce that the artificial game $\Gamma_{s}^v(t)$ can be equivalently reformulated as the perturbed game $\Gamma_s(t)$ with the perturbation $\eta(t)=(\eta^i(t;\varpi^i):i\in N,\varpi^i\in W^i)$ defined by $\eta^i(t;\varpi^i)=tp^{0i}(\varpi^i)$. The strategy profile $\gamma^*$ is a Nash equilibrium of $\Gamma_{s}^v(t)$ whenever $y(\gamma^*,t)$ constitutes a Nash equilibrium in $\Gamma_s(t)$. This leads to the following theorem.
	
	\begin{theorem}\label{nfpe-thm-ev1}{\em
			Consider the sequence $\{y(\gamma^{*k},t_k)\}_{k=1}^{\infty}$, where each $\gamma^{*k}$ serves as a Nash equilibrium of $\Gamma_{s}^v(t_k)$ with $t_k\in(0,1]$ and $\lim_{k\to \infty}t_k=0$. Then every limit point of the sequence $\{\sigma(y(\gamma^{*k},t_k))\}_{k=1}^{\infty}$ yields a normal-form perfect equilibrium.
		}
	\end{theorem}
	To identify a unique starting point, we introduce a logarithmic term to the problem~(\ref{nfpe-opt-ev1}) that extends the artificial game $\Gamma_{s}^v(t)$ to $t\in(0,2]$, thereby giving rise to the following convex optimization problem,
	\begin{equation}
		\label{nfpe-opt-hlt1}
		\begin{array}{rl}
			\max\limits_{\gamma^i(t)} & (1-\theta(t))\sum\limits_{j\in M_i}\sum\limits_{a\in A(I^j_i)}\gamma^i(t;\varpi^i_{I^j_i}a)g^i(\varpi^i_{I^j_i}a,y(\hat\gamma^{-i},\rho(t)))\\
			&-\theta(t)\sum\limits_{j\in M_i}\sum\limits_{a\in A(I^j_i)}\gamma^{0i}(\varpi^i_{I^j_i}a)\ln\gamma^i(t;\varpi^i_{I^j_i}a)\\
			
			\text{s.t.} & \sum\limits_{a\in A(I^j_i)}\gamma^i(t;\varpi^i_{I^j_i}a)-(1 - \theta(t))\gamma^i(t;\varpi^i_{I^j_i})-\theta(t)\gamma^{0i}(\varpi^i_{I^j_i})=0,\;j\in M_i,\\
			
			& 0\le\gamma^i(t;\varpi^i_{I^j_i}a),\;(j,a)\in D_i.
		\end{array}
	\end{equation}
	Applying the optimality conditions to the problem~(\ref{nfpe-opt-hlt1}), together with $\hat\gamma=\gamma(t)$, yields a polynomial equilibrium system. This system is specified by~(\ref{nfpe-eqt-hlt1}) for $t\in (1,2]$ and is equivalent to~(\ref{nfpe-eqt-ev1}) for $t\in (0,1]$.  
	\begin{equation}\label{nfpe-eqt-hlt1}\begin{array}{l}
			(1-\theta(t))g^i(\varpi^i_{I^j_i}a,y(\gamma^{-i}(t),\rho(t)))+ \lambda^i(\varpi^i_{I^j_i}a)\\
			\hspace{3.7cm}-\nu^i_{I^j_i}+(1 - \theta(t))\zeta^i_{I^j_i}(a)=0,\;i\in N,j\in M_i,a\in A(I^j_i),\\
			\sum\limits_{a\in A(I^j_i)}\gamma^i(t;\varpi^i_{I^j_i}a)-(1 - \theta(t))\gamma^i(t;\varpi^i_{I^j_i})-\theta(t)\gamma^{0i}(\varpi^i_{I^j_i})=0,\;j\in M_i,\\
			\gamma^i(t;\varpi^i_{I^j_i}a) \lambda^i(\varpi^i_{I^j_i}a)=\theta(t)\gamma^{0i}(\varpi^i_{I^j_i}a),\;0<\gamma^i(t;\varpi^i_{I^j_i}a),\;i\in N,j\in M_i,a\in A(I^j_i).
		\end{array}
	\end{equation}
	
	In order to overcome the non-differentiability induced by the boundary constraint conditions when $t\in (0,1]$, we implement a variable substitution. For $x=(x^i(\varpi^i_{I^j_i}a):i\in N,j\in M_i,a\in A(I^j_i))\in \mathbb{R}^{n_0}$, we define $\gamma(x,t)=(\gamma^i(x,t;\varpi^i_{I^j_i}a):i\in N,j\in M_i,a\in A(I^j_i))$ and $\lambda(x,t)=(\lambda^i(x,t;\varpi^i_{I^j_i}a):i\in N,j\in M_i,a\in A(I^j_i))$,  where
	\begin{equation*}\label{nfpe-eqt-hlt2}\begin{array}{l}
			\gamma^i(x,t;\varpi^i_{I^j_i}a)=\psi_1(x^i(\varpi^i_{I^j_i}a),\theta(t)^{1/\kappa_0}; \gamma^{0i}(\varpi^i_{I^j_i}a)^{1/\kappa_0},\kappa_0)\\
			\lambda^i(x,t;\varpi^i_{I^j_i}a)=\psi_2(x^i(\varpi^i_{I^j_i}a),\theta(t)^{1/\kappa_0};\gamma^{0i}(\varpi^i_{I^j_i}a)^{1/\kappa_0},\kappa_0),\;i\in N,j\in M_i,a\in A(I^j_i),
	\end{array}\end{equation*}
    and $\gamma^i(x,t;\emptyset)=1$. Consequently, we observe that
	$\gamma^i(x,t;\varpi^i_{I^j_i}a)\lambda^i(x,t;\varpi^i_{I^j_i}a)=\theta(t)\gamma^{0i}(\varpi^i_{I^j_i}a)$ holds for $i\in N,j\in M_i,a\in A(I^j_i)$. Substituting $\gamma(x,t)$ and $\lambda(x,t)$ into the system (\ref{nfpe-eqt-hlt1}) for $\gamma(t)$ and $\lambda$ and subtracting the expression $c(t)(1-\theta(t))\alpha$ from the obtained system, we reach the system (\ref{nfpe-eqt-hlt3}),
	\begin{equation}\label{nfpe-eqt-hlt3}\begin{array}{l}
			(1-\theta(t))g^i(\varpi^i_{I^j_i}a,y(\gamma^{-i}(x,t),\rho(t)))+\lambda^i(x,t;\varpi^i_{I^j_i}a)\\
			\hspace{0.4cm}-\nu^i_{I^j_i}+(1-\theta(t))\zeta^i_{I^j_i}(a)-c(t)(1-\theta(t))\alpha(\varpi^i_{I^j_i}a)=0,\;i\in N,j\in M_i,a\in A(I^j_i),\\
			\sum\limits_{a\in A(I^j_i)}\gamma^i(x,t;\varpi^i_{I^j_i}a)-(1 - \theta(t))\gamma^i(x,t;\varpi^i_{I^j_i})-\theta(t)\gamma^{0i}(\varpi^i_{I^j_i})=0,\;i\in N,j\in M_i.
		\end{array}
	\end{equation}
	At $t=2$, the system possesses a unique solution represented by $(x^*(2),\nu^*(2))$. The components are expressed as  $x^{*i}(2;\varpi^i_{I^j_i}a)=\gamma^{0i}(\varpi^i_{I^j_i}a)^{1/\kappa_0}-1\text{ for }i\in N,j\in M_i,a\in A(I^j_i)$, and $\nu^{*i}_{I^j_i}(2)=1\text{ for }i\in N,j\in M_i$.
	
	Define $\widetilde{\mathscr{P}}_V=\{(x,\nu,t)|(x,\nu,t)\text{ satisfies the system~(\ref{nfpe-eqt-hlt3}) with } 0<t\leq 2\}$, and let $\mathscr{P}_V$ denote its closure. By employing the transversality theorem alongside the implicit function theorem, it follows that for almost any $\alpha\in\mathbb{R}^{n_0}$ with sufficiently small $||\alpha||$, there exists a smooth path within $\mathscr{P}_V$. This path originates at $(x^*(2),\nu^*(2),2)$ when $t=2$, and as $t$ tends to zero, the limit of $\sigma(y(\gamma(x^*(t),t),\rho(t)))$ with $(x^*(t),t)$ along this smooth path, converges to a normal-form perfect equilibrium.
	\subsection{Harsanyi's logarithmic tracing procedure in sequence form}\label{nfpe-sec-htp2}
	Harsanyi's logarithmic tracing procedure aims to approximate the piecewise equilibrium path from the linear tracing method with a smooth path, thereby improving its efficiency. This section extends the procedure to the sequence form.
	
	Let $\varepsilon_0$ be a positive constant and $\delta=(\delta^{i}(\varpi^i):i\in N,\varpi^i\in W^i)$ denote the centroid realization plan profile with $\delta^i(\varpi^i_{I^j_i}a)=\delta^i(\varpi^i_{I^j_i})/|A(I^j_i)|$. By expanding the influence interval of the logarithmic term to $(0,2]$ in the problem (\ref{nfpe-opt-hlt1}), we derive a new artificial game $\Gamma_s^h(t)$, in which each player $i\in N$ finds their optimal strategy by solving the convex optimization problem,
	\begin{equation}
		\label{nfpe-opt-hlg1}
		\begin{array}{rl}
			\max\limits_{\gamma^i(t)} & (1-\theta(t))\sum\limits_{j\in M_i}\sum\limits_{a\in A(I^j_i)}\gamma^i(t;\varpi^i_{I^j_i}a)g^i(\varpi^i_{I^j_i}a,y(\hat\gamma^{-i},\rho(t)))\\
			&+\sum\limits_{j\in M_i}\sum\limits_{a\in A(I^j_i)}(\theta(t)\gamma^{0i}(\varpi^i_{I^j_i}a)+c(t)(1-\theta(t))\varepsilon_0\delta^i(\varpi^i_{I^j_i}a))\ln\gamma^i(t;\varpi^i_{I^j_i}a)\\
			
			\text{s.t.} & \sum\limits_{a\in A(I^j_i)}\gamma^i(t;\varpi^i_{I^j_i}a)-(1 - \theta(t))\gamma^i(t;\varpi^i_{I^j_i})-\theta(t)\gamma^{0i}(\varpi^i_{I^j_i})=0,\;j\in M_i.
		\end{array}
	\end{equation}
	When $t\in (0,1]$, each player $i$ is incentivized to adjust their strategy closer to the centroid strategy $\delta^i$, and (\ref{nfpe-opt-hlg1}) increasingly approximates (\ref{nfpe-opt-hlt1}) as $\varepsilon_0$ goes to $0$. Applying the optimality conditions to the problem~(\ref{nfpe-opt-hlg1}) and enforcing $\hat\gamma = \gamma(t)$ produces the following polynomial equilibrium system,
	\begin{equation}\label{nfpe-eqt-hlg1}\begin{array}{l}
			(1-\theta(t))g^i(\varpi^i_{I^j_i}a,y(\gamma^{-i}(t),\rho(t)))+ \lambda^i(\varpi^i_{I^j_i}a)-\nu^i_{I^j_i}\\
			\hspace{3.5cm}+(1 - \theta(t))\zeta^i_{I^j_i}(a)=0,\;i\in N,j\in M_i,a\in A(I^j_i),\\
			\sum\limits_{a\in A(I^j_i)}\gamma^i(t;\varpi^i_{I^j_i}a)-(1 - \theta(t))\gamma^i(t;\varpi^i_{I^j_i})-\theta(t)\gamma^{0i}(\varpi^i_{I^j_i})=0,\;i\in N,j\in M_i,\\
			\gamma^i(t;\varpi^i_{I^j_i}a) \lambda^i(\varpi^i_{I^j_i}a)=\theta(t)\gamma^{0i}(\varpi^i_{I^j_i}a)+c(t)(1-\theta(t))\varepsilon_0\delta^i(\varpi^i_{I^j_i}a),\\
			\hspace{3cm}0<\gamma^i(t;\varpi^i_{I^j_i}a),\;i\in N,j\in M_i,a\in A(I^j_i).
		\end{array}
	\end{equation}
	For $x=(x^i(\varpi^i_{I^j_i}a):i\in N,j\in M_i,a\in A(I^j_i))\in \mathbb{R}^{n_0}$, we define $\gamma(x,t)=(\gamma^i(x,t;\varpi^i_{I^j_i}a):i\in N,j\in M_i,a\in A(I^j_i))$ and $\lambda(x,t)=(\lambda^i(x,t;\varpi^i_{I^j_i}a):i\in N,j\in M_i,a\in A(I^j_i))$,  where
	\begin{equation}\label{/}\begin{array}{l}
			\gamma^i(x,t;\varpi^i_{I^j_i}a)=\psi_1(x^i(\varpi^i_{I^j_i}a),(\theta(t)\gamma^{0i}(\varpi^i_{I^j_i}a)+c(t)(1-\theta(t))\varepsilon_0\delta^i)^{1/\kappa_0};1,\kappa_0),\\
			\lambda^i(x,t;\varpi^i_{I^j_i}a)=\psi_2(x^i(\varpi^i_{I^j_i}a),(\theta(t)\gamma^{0i}(\varpi^i_{I^j_i}a)+c(t)(1-\theta(t))\varepsilon_0\delta^i)^{1/\kappa_0};1,\kappa_0).\\
	\end{array}\end{equation}
	Substituting $\gamma(x,t)$ and $\lambda(x,t)$ into the system (\ref{nfpe-eqt-hlg1}) for $\gamma(t)$ and $\lambda$, and subsequently subtracting the expression $c(t)(1-\theta(t))\alpha$ from the equivalent system, we obtain the system (\ref{nfpe-eqt-hlg3}),
	\begin{equation}\label{nfpe-eqt-hlg3}\begin{array}{l}
			(1-\theta(t))g^i(\varpi^i_{I^j_i}a,y(\gamma^{-i}(x,t),\rho(t)))+ \lambda^i(x,t;\varpi^i_{I^j_i}a)\\
			\hspace{0.6cm}-\nu^i_{I^j_i}+(1 - \theta(t))\zeta^i_{I^j_i}(a)-c(t)(1-\theta(t))\alpha(\varpi^i_{I^j_i}a)=0,\;i\in N,,j\in M_i,a\in A(I^j_i),\\
			\sum\limits_{a\in A(I^j_i)}\gamma^i(x,t;\varpi^i_{I^j_i}a)-(1 - \theta(t))\gamma^i(x,t;\varpi^i_{I^j_i})-\theta(t)\gamma^{0i}(\varpi^i_{I^j_i})=0,\;i\in N,j\in M_i.
		\end{array}
	\end{equation}
	At $t=2$, this system has a unique solution, $(x^*(2),\nu^*(2))$, which is identical to the sole solution of the system (\ref{nfpe-eqt-hlt3}). Furthermore, the same conclusion as in the previous subsection can still be drawn, namely, that a distinguished smooth path exists in the solution set of the system (\ref{nfpe-eqt-hlg3}), originating from $(x^*(2),\nu^*(2))$ at $t=2$, and converging to a normal-form perfect equilibrium as $t$ approaches zero.
	
	\section{Numerical performance}\label{nfpe-sec-prm6}
	In this section, we present a set of numerical experiments aimed at evaluating the effectiveness and efficiency of the proposed methods. Our investigation centers on three primary aspects: \begin{itemize}
		\item The ability of our algorithm converging to a more stable normal-form perfect equilibrium in extensive-form games, especially in those possessing unstable extensive-form perfect equilibria.
		\item The effectiveness of our methods in addressing complex multi-player, multi-action games.
		\item A comparative analysis of the three methods in handling large-scale games. \end{itemize}
	To achieve these objectives, we employed the predictor-corrector method to numerically trace the smooth paths defined by the systems (\ref{nfpe-eqt-log4}), (\ref{nfpe-eqt-hlt3}), and (\ref{nfpe-eqt-hlg3}), respectively referred to as LOGB, HLTP, and HLOG ($\varepsilon_0=1$). In addition, to enhance the smoothness of the paths near $t=1$, we adopted the variable-substitution adjustment technique proposed in~\cite{houSequenceformDifferentiablePathfollowing2025}. During the tracing procedure, each iteration comprised a predictor to approximate the next solution and a corrector to refine this approximation for improved accuracy. Detailed implementation guidelines can be found in Allgower and Georg~\cite{AllgowerNumericalContinuationMethods1990} and Rheinboldt~\cite{rheinboldtNumericalContinuationMethods2000}. The adopted parameter settings including a predictor step size of $0.05t^{0.3}$ and a corrector accuracy of $0.5t^{0.3}$. The successful termination criterion $t<10^{-4}$ was applied, with failure occurring when the number of iterations or computational time surpassed predefined limits. All computations were conducted on a Windows Server 2016 Standard with an Intel(R) Xeon(R) CPU E5-2650 v4 @ 2.20GHz (2 processors) and 128GB of RAM.
	\begin{example}\label{nfpe-exm-num1} {\em In this example, we examine two extensive-form games, depicted in Figures~\ref{nfpe-fig-gam1}--\ref{nfpe-fig-gam2}, to validate the capability of our methods in converging to normal-form perfect equilibria when existing unstable extensive-form perfect equilibria. For the first game, the sole normal-form perfect equilibrium is given by $\sigma^1(ac)=1,\sigma^2(A)=1$. In the second game, the unique normal-form perfect equilibrium is $\sigma^1(L_1L_2L_3)=1,\sigma^2(A)=1,\sigma^3(C)=1$. However, additional unstable extensive-form perfect equilibria exist in both games. Using our methods with randomly chosen starting points from the feasible region, we obtain smooth paths that converge to the normal-form perfect equilibria, as illustrated in Figures~\ref{NFPE-fig-pth1}--\ref{NFPE-fig-pth6}.
			\begin{figure}[H]
				\centering
				\begin{minipage}[b]{0.41\textwidth}
					\centering
					\includegraphics[width=0.9\textwidth]{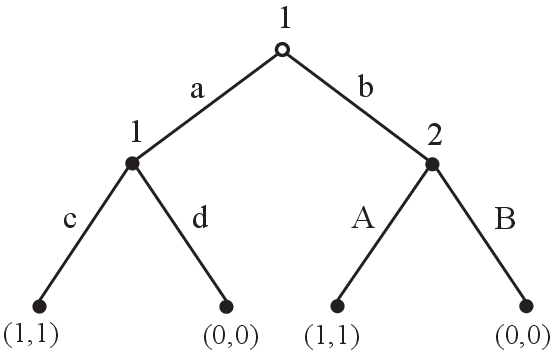}
					\caption{\label{nfpe-fig-gam1}{\small An Extensive-Form Game from van Damme~\cite{vandammeRelationPerfectEquilibria1984}}}
				\end{minipage}\hfill
				\begin{minipage}[b]{0.57\textwidth}
					\centering
					\includegraphics[width=0.9\textwidth]{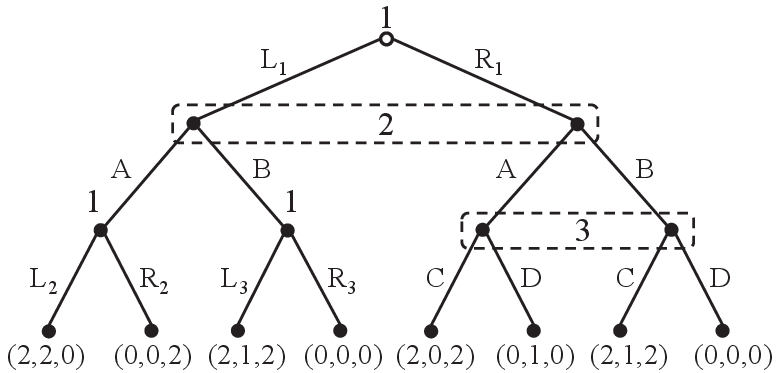}
					\caption{\label{nfpe-fig-gam2}{\small An Extensive-Form Game}}
				\end{minipage}
		\end{figure}}
	\end{example}
	\begin{figure}[htp]
		\centering
		\begin{minipage}[b]{0.49\textwidth}
			\centering
			\includegraphics[width=1\textwidth, height=0.20\textheight]{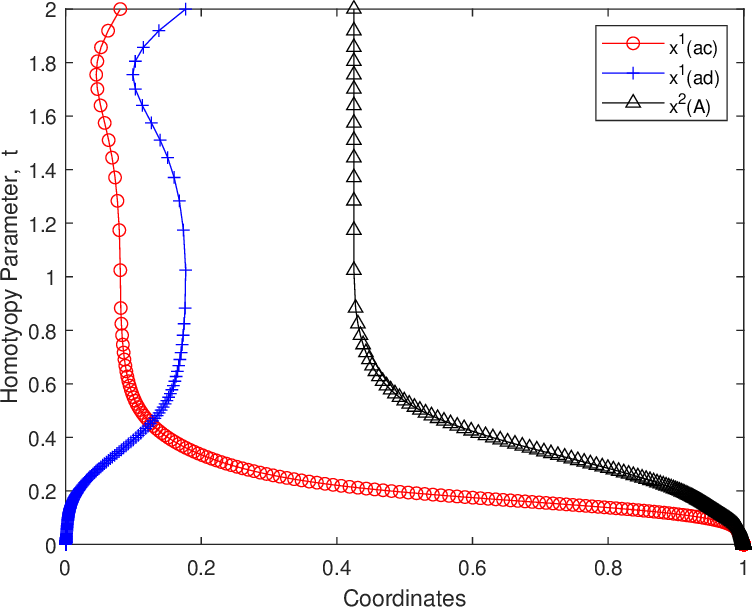}
			\caption{\label{NFPE-fig-pth1}{\footnotesize Path of Mixed Strategies Generated by LOGB for the Game in Fig.~\ref{nfpe-fig-gam1}}}\end{minipage}\hfill
		\begin{minipage}[b]{0.49\textwidth}
			\centering
			\includegraphics[width=1\textwidth, height=0.20\textheight]{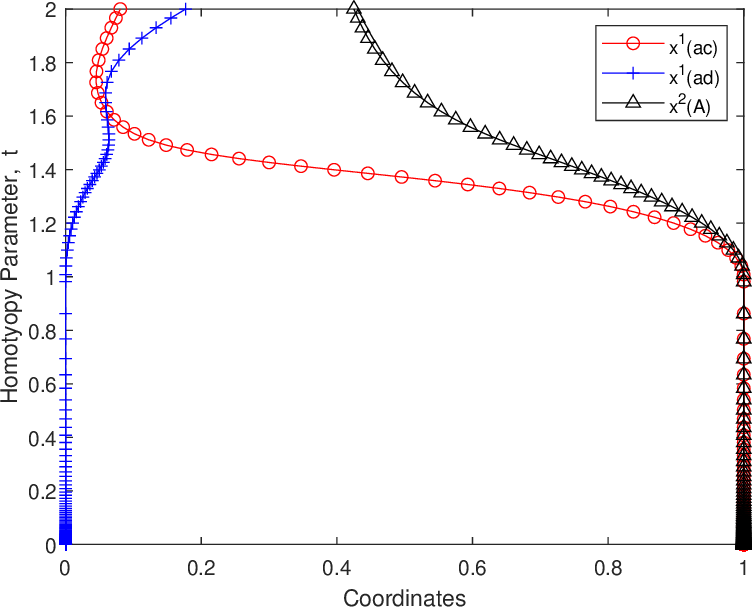}
			\caption{\label{NFPE-fig-pth2}{\footnotesize Path of Mixed Strategies Generated by HLTP for the Game in Fig.~\ref{nfpe-fig-gam1}}} \end{minipage}
	\end{figure}
	\begin{figure}[htp]
		\centering
		\begin{minipage}{0.49\textwidth}
			\centering
			\includegraphics[width=1\textwidth, height=0.20\textheight]{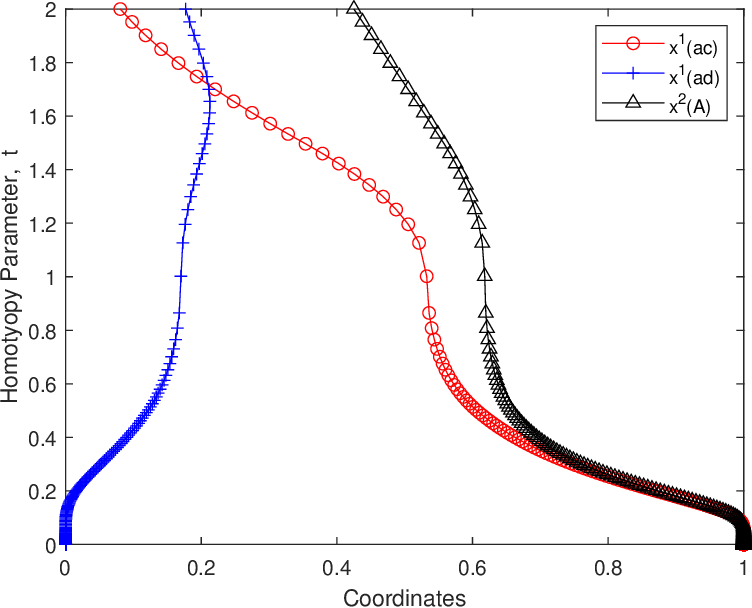}
			\caption{\label{NFPE-fig-pth3}{\footnotesize Path of Mixed Strategies Generated by HLOG for the Game in Fig.~\ref{nfpe-fig-gam1}}}\end{minipage}\hfill
		\begin{minipage}{0.49\textwidth}
			\centering
			\includegraphics[width=1\textwidth, height=0.20\textheight]{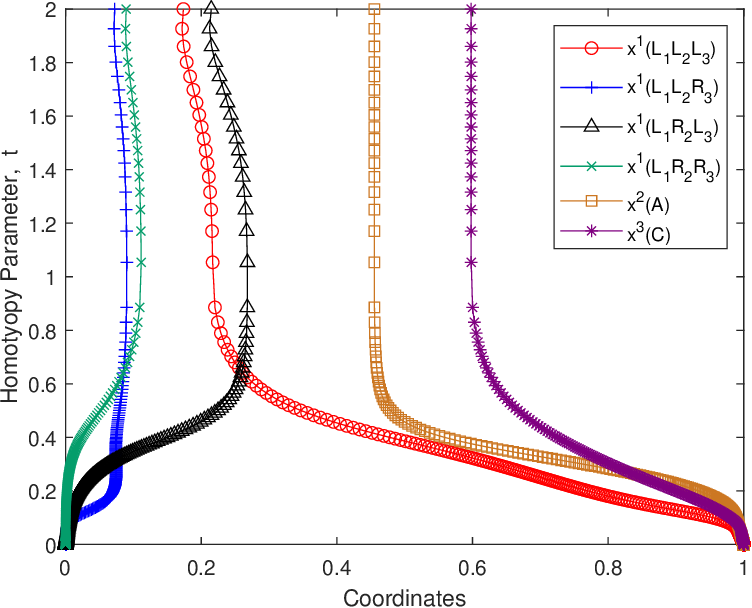}
			\caption{\label{NFPE-fig-pth4}{\footnotesize Path of Mixed Strategies Generated by LOGB for the Game in Fig.~\ref{nfpe-fig-gam2}}} \end{minipage}
	\end{figure}
	\begin{figure}[htp]
		\centering
		\begin{minipage}{0.49\textwidth}
			\centering
			\includegraphics[width=1\textwidth, height=0.20\textheight]{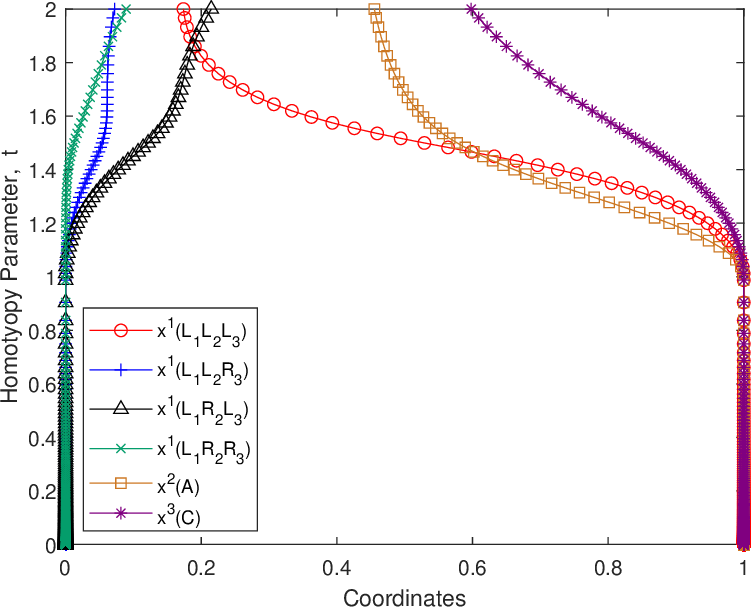}
			\caption{\label{NFPE-fig-pth5}{\footnotesize Path of Mixed Strategies Generated by HLTP for the Game in Fig.~\ref{nfpe-fig-gam2}}}\end{minipage}\hfill
		\begin{minipage}{0.49\textwidth}
			\centering
			\includegraphics[width=1\textwidth, height=0.20\textheight]{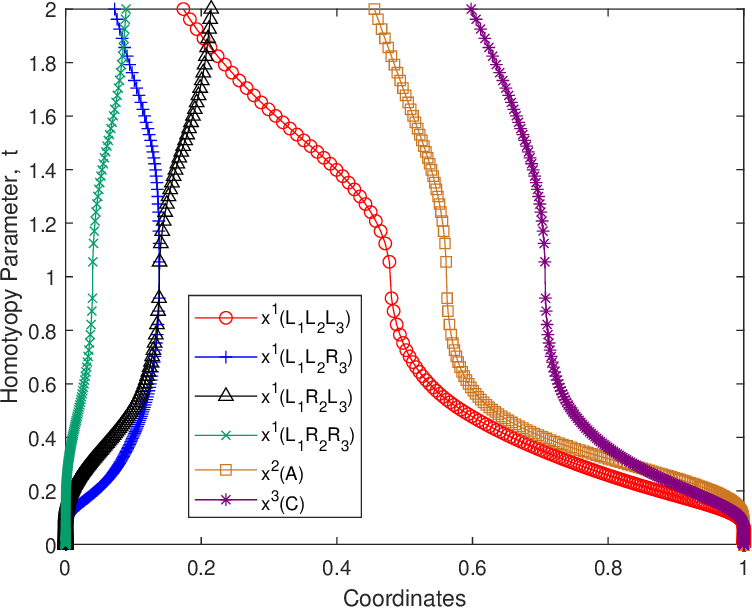}
			\caption{\label{NFPE-fig-pth6}{\footnotesize Path of Mixed Strategies Generated by HLOG for the Game in Fig.~\ref{nfpe-fig-gam2}}} \end{minipage}
	\end{figure}
	
	\begin{example}\label{nfpe-exm-num2} {\em In this example, we consider two extensive-form games, depicted in Figures~\ref{nfpe-fig-gam3}--\ref{nfpe-fig-gam4}, to evaluate the effectiveness of our methods in solving complex multi-player, multi-action games. The starting point is randomly chosen from the feasible region, and the corresponding paths in mixed strategies are illustrated in Figures~\ref{NFPE-fig-pth7}--\ref{NFPE-fig-pth12}. These paths successfully converge to a normal-form perfect equilibrium.
			\begin{figure}[H]
				\centering
				\begin{minipage}[b]{0.42\textwidth}
					\centering
					\includegraphics[width=0.9\textwidth]{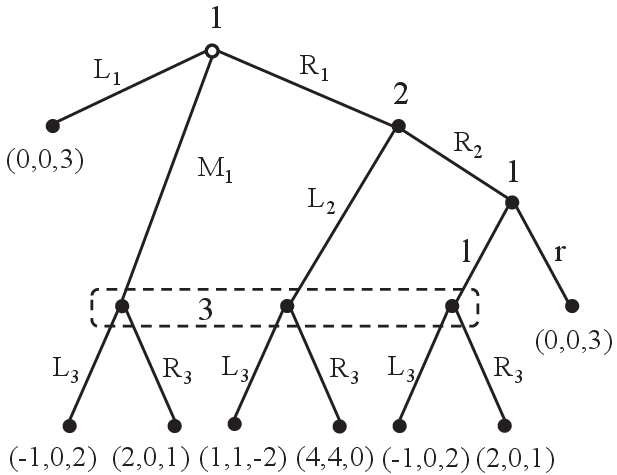}
					\caption{\label{nfpe-fig-gam3}{\small An Extensive-Form Game from Mas-Colell et al.~\cite{andreuMicroeconomicTheory1995}}}
				\end{minipage}
				\begin{minipage}[b]{0.56\textwidth}
					\centering
					\includegraphics[width=0.9\textwidth]{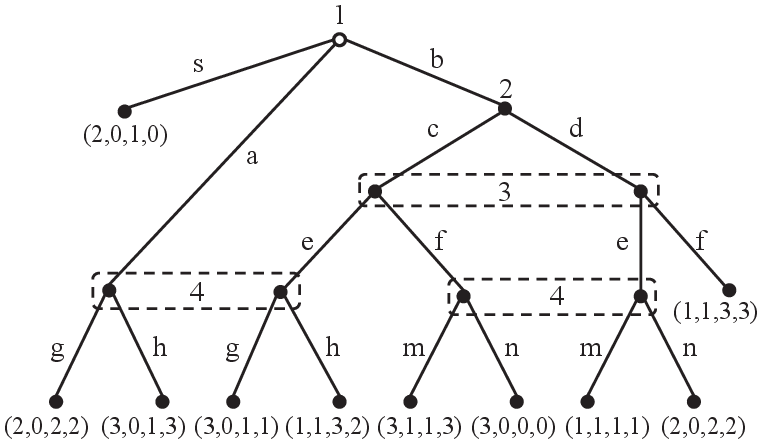}
					\caption{\label{nfpe-fig-gam4}{\small An Extensive-Form Game from Bonanno~\cite{bonanno2015game}}}
				\end{minipage}\hfill
			\end{figure}}
	\end{example}
	\begin{figure}[htp]
		\centering
		\begin{minipage}{0.49\textwidth}
			\centering
			\includegraphics[width=1\textwidth, height=0.20\textheight]{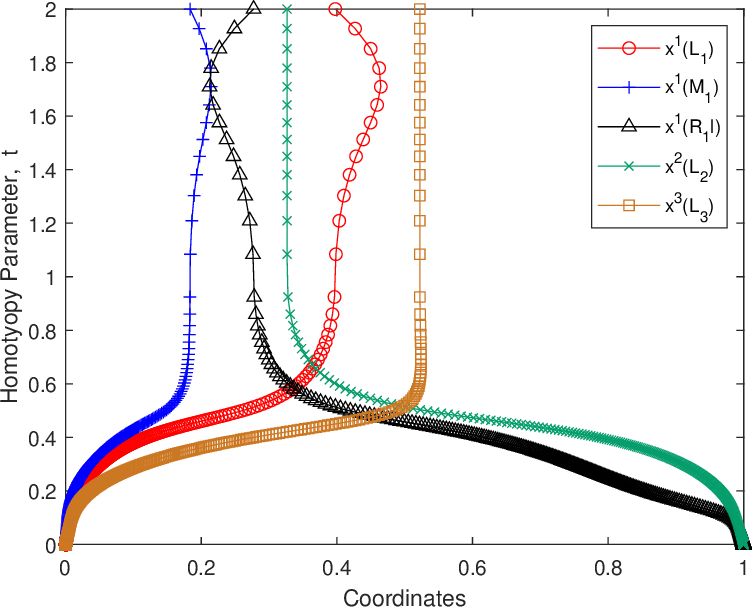}
			\caption{\label{NFPE-fig-pth7}{\footnotesize Path of Mixed Strategies Generated by LOGB for the Game in Fig.~\ref{nfpe-fig-gam3}}}\end{minipage}\hfill
		\begin{minipage}{0.49\textwidth}
			\centering
			\includegraphics[width=1\textwidth, height=0.20\textheight]{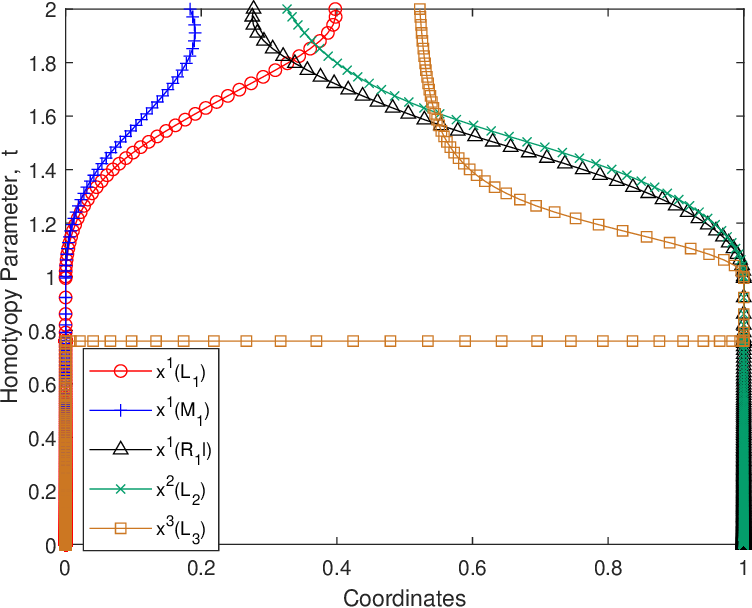}
			\caption{\label{NFPE-fig-pth8}{\footnotesize Path of Mixed Strategies Generated by HLTP the Game in Fig.~\ref{nfpe-fig-gam3}}} \end{minipage}
	\end{figure}
	\begin{figure}[htp]
		\centering
		\begin{minipage}{0.49\textwidth}
			\centering
			\includegraphics[width=1\textwidth, height=0.2\textheight]{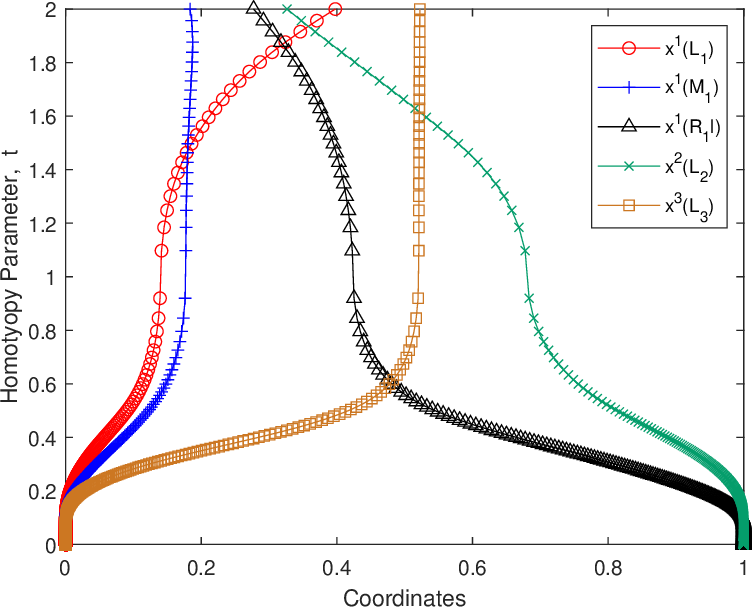}
			\caption{\label{NFPE-fig-pth9}{\footnotesize Path of Mixed Strategies Generated by HLOG the Game in Fig.~\ref{nfpe-fig-gam3}}}\end{minipage}\hfill
		\begin{minipage}{0.49\textwidth}
			\centering
			\includegraphics[width=1\textwidth, height=0.2\textheight]{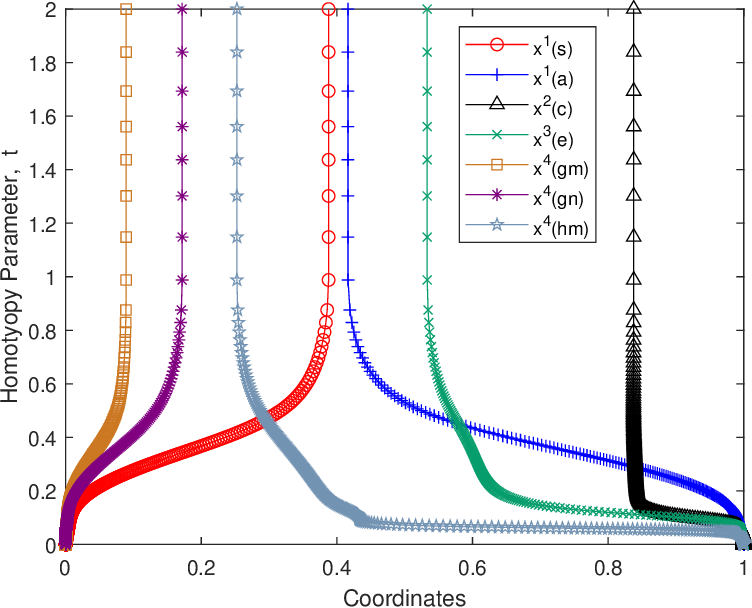}
			\caption{\label{NFPE-fig-pth10}{\footnotesize Path of Mixed Strategies Generated by LOGB the Game in Fig.~\ref{nfpe-fig-gam4}}} \end{minipage}
	\end{figure}
	\begin{figure}[htp]
		\centering
		\begin{minipage}{0.49\textwidth}
			\centering
			\includegraphics[width=1\textwidth, height=0.2\textheight]{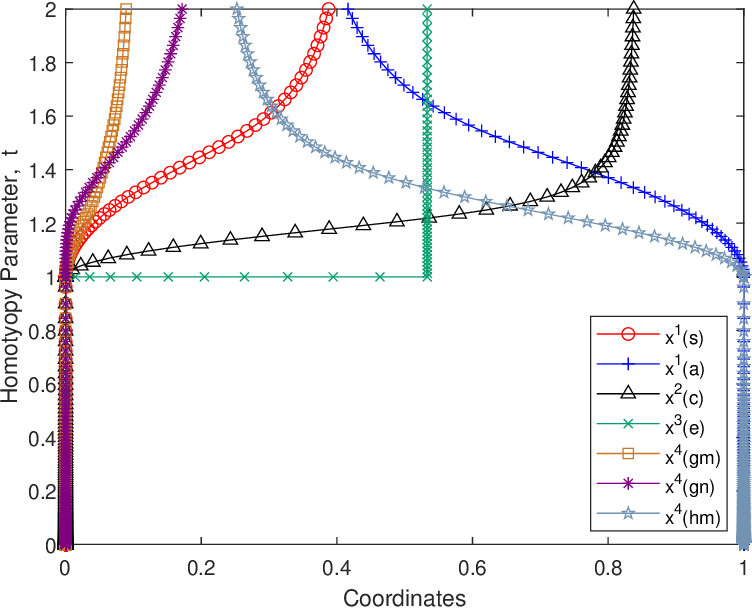}
			\caption{\label{NFPE-fig-pth11}{\footnotesize Path of Mixed Strategies Generated by HLTP the Game in Fig.~\ref{nfpe-fig-gam4}}}\end{minipage}\hfill
		\begin{minipage}{0.49\textwidth}
			\centering
			\includegraphics[width=1\textwidth, height=0.2\textheight]{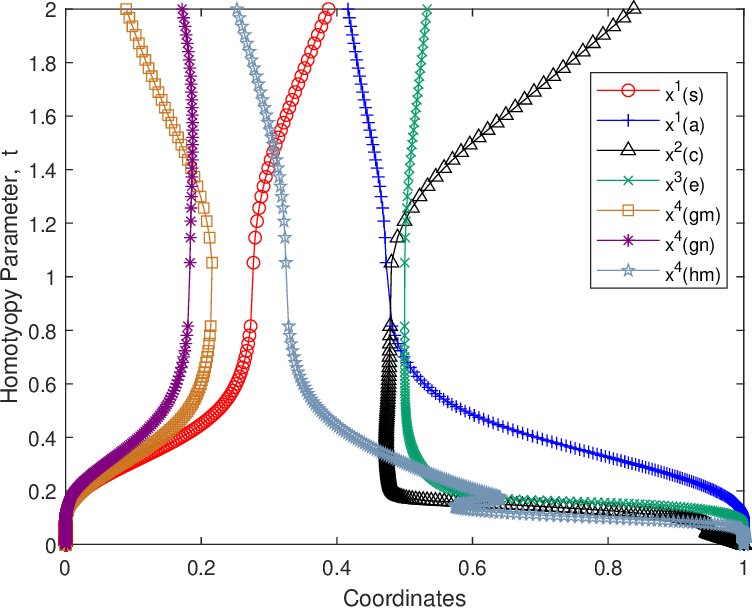}
			\caption{\label{NFPE-fig-pth12}{\footnotesize Path of Mixed Strategies Generated by HLOG the Game in Fig.~\ref{nfpe-fig-gam4}}} \end{minipage}
	\end{figure}
	\begin{example}\label{nfpe-exm-num3} {\em To compare the convergence performance of our methods, we employ two structurally distinct types of random extensive-form games, as shown in Figures~\ref{nfpe-fig-gam5}--\ref{nfpe-fig-gam6}. Both game types are parameterized by the number of players ($n$), the maximum historical depth ($\mathcal{L}$), and the number of allowable actions per information set ($\mathcal{A}$). In these games, players act cyclically, with the terminal payoffs determined by random integers uniformly distributed between $-10$ and $10$. A detailed explanation of the two game types is provided below.
	\begin{itemize}
		\item \textbf{Type 1:} As shown in Figures~\ref{nfpe-fig-gam5}, histories are classified into the same information set only when they diverge in the final actions taken. Moreover, all terminal histories exhibit an identical length.
		\item \textbf{Type 2:} As represented in Figures~\ref{nfpe-fig-gam6}, this structural configuration is commonly found in the literature. For odd-indexed players, each information set consists of a single history. In contrast, for even-indexed players, histories are grouped into the same information set only when they share an identical corresponding sequence. The probability that player $0$ chooses each of the available actions is equal, and the total number of actions is fixed at $3$, without loss of generality.
	\end{itemize}
	Since the number of players does not directly impact the game size, we set $n=3$ for Type 1 games and $n=4$ for Type 2 games, adjusting the other two parameters to control the game size. To conduct a thorough comparative evaluation of the three path-following methods, $20$ random games with distinct payoffs were generated and solved for each parameter configuration $(\mathcal{L}, \mathcal{A})$ across both types of games. For each game, all three methods utilized the same randomly generated starting point, and the parameters of the predictor-corrector algorithm remained consistent throughout the entire experiment.
	
	The numerical results in Tables~\ref{t2} and \ref{t3} show that the LOGB method consistently outperforms HLTP and HLOG in terms of numerical stability, efficiency, and scalability. LOGB achieves a $0\%$ failure rate across all tested games, requires fewer iterations, and has the shortest computational time overall. Although HLTP sometimes converges faster in iteration numbers for small-scale games, it suffers from high failure rates and poor scalability. HLOG performs better than HLTP in terms of stability but is significantly slower and less efficient. Overall, LOGB provides the most reliable and efficient performance, making it the preferred method for computing normal-form perfect equilibria across a wide range of game sizes.
		\begin{figure}[H]
				\centering
				\begin{minipage}[b]{0.55\textwidth}
					\centering
					\includegraphics[width=0.95\textwidth]{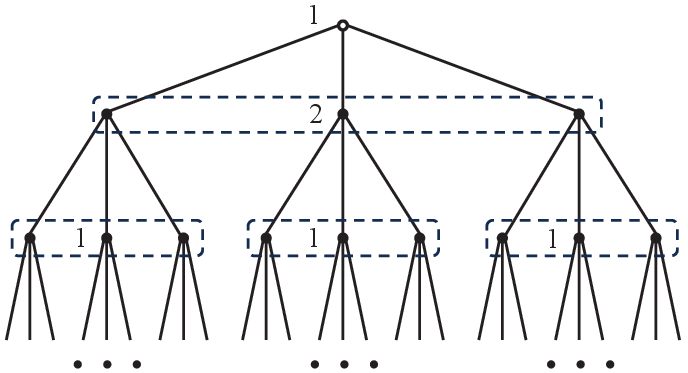}
					\caption{\label{nfpe-fig-gam5}{\small A Random Extensive-Form Game of Type 1}}
				\end{minipage}\hfill
				\begin{minipage}[b]{0.43\textwidth}
					\centering
					\includegraphics[width=0.95\textwidth]{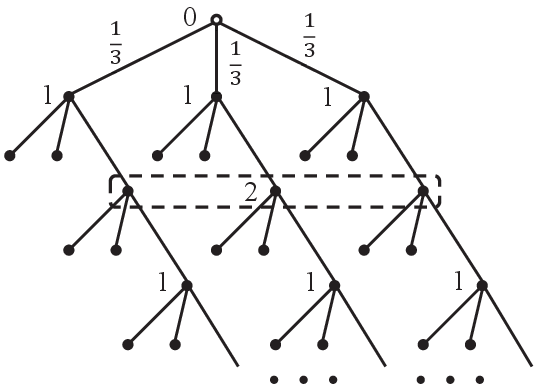}
					\caption{\label{nfpe-fig-gam6}{\small A Random Extensive-Form Game of Type 2}}
				\end{minipage}
		\end{figure}}
	\end{example}
	\begin{table}[htbp]\centering\renewcommand\arraystretch{0.95}\setlength{\tabcolsep}{5pt}
		\caption{Numerical Comparisons for the Game in Fig.~\ref{nfpe-fig-gam5}}\label{t2}
		\begin{tabular*}{\textwidth}{@{\extracolsep\fill}>{\rowmac}c>{\rowmac}l>{\rowmac}l>{\rowmac}l>{\rowmac}l>{\rowmac}l>{\rowmac}l>{\rowmac}l>{\rowmac}l>{\rowmac}l>{\rowmac}l<{\clearrow}}\toprule
			\multirow{2}*{$(\mathcal{L},\mathcal{A})$}&  & \multicolumn{3}{l}{Iteration Numbers} & \multicolumn{3}{l}{Computational Time} & \multicolumn{3}{l}{Failure Rates} \\\cmidrule(r){3-5}\cmidrule(r){6-8}\cmidrule(r){9-11}
			& & LOGB & HLTP & HLOG & LOGB & HLTP & HLOG & LOGB & HLTP & HLOG\\\midrule
			$(5,2)$ & max & 807 & - & 527 & 120.9 & - & 41.7 & 0\% & 5\% & 0\%\\
			& min & 165 & 226 & 232 & 7.2 & 8.2 & 8.7 & & & \\
			&\setrow{\bfseries} med & 192.5 & 465.5 & 259.5 & 8.1 & 24.4 & 10.4 & & &\\
			$(6,2)$ & max & 480 & - & 3657 & 107.5 & - & 658.6 & 0\% & 20\% & 0\%\\
			& min & 177 & 402 & 229 & 28.8 & 65.8 & 32.8 & & & \\
			&\setrow{\bfseries} med & 257.0 & 1590.0 & 316.0 & 40.3 & 362.4 & 45.8 & & &\\
			$(7,2)$ & max & 4299 & - & - & 3531.9 & - & - & 0\% & 70\% & 25\%\\
			& min & 224 & 543 & 305 & 165.8 & 536.2 & 217.4 & & & \\
			&\setrow{\bfseries} med & 286.0 & - & 419.5 & 211.5 & - & 290.3 & & &\\
			$(8,2)$ & max & - & - & - & - & - & - & 20\% & 95\% & 25\%\\
			& min & 276 & 766 & 349 & 1482.9 & 5327.8 & 1544.6 & & & \\
			&\setrow{\bfseries} med & 427.0 & - & 565.5 & 2090.6 & - & 2661.2 & & &\\
			$(4,3)$ & max & 492 & - & 3356 & 44.3 & - & 364.6 & 0\% & 25\% & 0\%\\
			& min & 200 & 341 & 234 & 19.8 & 50.6 & 24.1 & & & \\
			&\setrow{\bfseries} med & 243.0 & 2120.0 & 306.5 & 24.7 & 385.8 & 30.4 & & &\\
			$(4,4)$ & max & 694 & - & - & 587.1 & - & - & 0\% & 40\% & 5\%\\
			& min & 248 & 374 & 269 & 169.5 & 372.1 & 184.2 & & & \\
			&\setrow{\bfseries} med & 298.5 & 3243.0 & 371.0 & 213.9 & 3329.6 & 251.9 & & &\\
			$(4,5)$ & max & - & - & - & - & - & - & 25\% & 95\% & 30\%\\
			& min & 301 & 679 & 376 & 1052.9 & 3545.4 & 1225.1 & & & \\
			&\setrow{\bfseries} med & 471.0 & - & 635.5 & 2236.7 & - & 2385.1 & & &\\
			$(4,6)$ & max & - & - & - & - & - & - & 80\% & 100\% & 95\%\\
			& min & 323 & - & 360 & 4975.7 & - & 5305.4 & & & \\
			&\setrow{\bfseries} med & - & - & - & - & - & - & & &\\
			\bottomrule
		\end{tabular*}
	\end{table}
	\begin{table}[htbp]\centering\renewcommand\arraystretch{0.95}\setlength{\tabcolsep}{5pt}
		\caption{Numerical Comparisons for the Game in~Fig.~\ref{nfpe-fig-gam6}}\label{t3}
		\begin{tabular*}{\textwidth}{@{\extracolsep\fill}>{\rowmac}c>{\rowmac}l>{\rowmac}l>{\rowmac}l>{\rowmac}l>{\rowmac}l>{\rowmac}l>{\rowmac}l>{\rowmac}l>{\rowmac}l>{\rowmac}l<{\clearrow}}\toprule
			\multirow{2}*{$(\mathcal{L},\mathcal{A})$}& & \multicolumn{3}{l}{Iteration Numbers} & \multicolumn{3}{l}{Computational Time} & \multicolumn{3}{l}{Failure Rates}\\\cmidrule(r){3-5}\cmidrule(r){6-8}\cmidrule(r){9-11}
			& & LOGB & HLTP & HLOG & LOGB & HLTP & HLOG & LOGB & HLTP & HLOG\\\midrule
			$(10,2)$ & max & 502 & 5436 & 1613 & 46.9 & 1023.0 & 287.4 & 0\% & 0\% & 0\%\\
                & min & 91 & 121 & 138 & 7.3 & 7.2 & 9.2 & & & \\
                &\setrow{\bfseries} med & 147.0 & 190.0 & 216.5 & 11.9 & 19.2 & 14.6 & & &\\
			$(20,2)$ & max & 813 & - & 2452 & 1000.3 & - & 606.5 & 0\% & 5\% & 0\%\\
			& min & 113 & 115 & 204 & 40.1 & 34.9 & 56.6 & & & \\
			&\setrow{\bfseries} med & 160.0 & 186.0 & 361.0 & 51.6 & 100.7 & 101.4 & & &\\
			$(30,2)$ & max & 483 & 5394 & 5075 & 843.2 & 4070.3 & 2896.1 & 0\% & 0\% & 0\%\\
			& min & 171 & 164 & 210 & 128.6 & 142.0 & 141.7 & & & \\
			&\setrow{\bfseries} med & 258.5 & 284.0 & 537.5 & 206.9 & 389.2 & 379.3 & & &\\
			$(40,2)$ & max & 473 & - & - & 868.7 & - & - & 0\% & 5\% & 5\%\\
			& min & 293 & 172 & 250 & 503.1 & 269.9 & 424.4 & & & \\
			&\setrow{\bfseries} med & 396.0 & 471.0 & 1217.0 & 674.4 & 1234.0 & 1728.4 & & &\\
			$(10,4)$ & max & 9600 & 1945 & 1435 & 3522.8 & 724.5 & 446.6 & 0\% & 0\% & 0\%\\
			& min & 143 & 81 & 190 & 47.5 & 32.6 & 55.5 & & & \\
			&\setrow{\bfseries} med & 294.0 & 171.5 & 358.0 & 96.9 & 83.6 & 101.2 & & &\\
			$(10,6)$ & max & 431 & - & 1164 & 377.3 & - & 891.6 & 0\% & 5\% & 0\%\\
			& min & 165 & 117 & 213 & 139.0 & 96.5 & 156.9 & & & \\
			&\setrow{\bfseries} med & 266.5 & 180.5 & 416.5 & 217.1 & 175.3 & 281.7 & & &\\
			$(10,8)$ & max & 1023 & - & 1097 & 1545.9 & - & 1326.0 & 0\% & 5\% & 0\%\\
			& min & 98 & 82 & 144 & 8.5 & 6.8 & 10.7 & & & \\
			&\setrow{\bfseries} med & 321.0 & 186.0 & 336.0 & 527.3 & 311.7 & 477.8 & & &\\
			$(10,10)$ & max & 711 & - & 2198 & 2108.6 & - & 4535.2 & 0\% & 5\% & 0\%\\
			& min & 293 & 97 & 224 & 770.9 & 353.6 & 565.0 & & & \\
			&\setrow{\bfseries} med & 346.5 & 178.0 & 489.0 & 934.5 & 511.3 & 1137.1 & & &\\
			\bottomrule
		\end{tabular*}
	\end{table}
	\section{Conclusion}\label{nfpe-sec-prm7}
	The sequence form's holistic property enables the development of a sequence-form characterization for normal-form refinements of Nash equilibria, and its compactness contributes to computational efficiency. Inspired by this, we have developed a sequence-form characterization of normal-form perfect equilibria for extensive-form games with perfect recall. Guided by this theoretical foundation, we have proposed three distinct sequence-form differential path-following methods for computing normal-form perfect equilibria and rigorously proved their convergence. These methods are underpinned by the construction of artificial games, where the first method incorporates logarithmic-barrier terms into the payoff functions, while the last two methods are derived by extending Harsanyi's linear and logarithmic tracing procedure to the sequence form, respectively. Each of the three methods allows for choosing the starting point within a specified range. The existence of smooth paths converging to normal-form perfect equilibria has been established by means of  theoretical analysis and numerical experiments. To evaluate the performance of these methods, we have constructed two distinct types of random games for comparative experiments. The experimental results further substantiate the effectiveness and efficiency of our methods. Future work could investigate the computation of other normal-form refinements of Nash equilibrium in $n$-player games, such as normal-form proper equilibrium.
	
	\section*{Acknowledgments}
	This work was partially supported by GRF: CityU 11306821 of Hong Kong SAR Government.
	
	\begin{appendices}
		\section{}\label{nfpe-appen-2}
		Let $(\gamma^*{(t)}, t, \lambda^*, \nu^*)$ be a solution to the system~(\ref{nfpe-eqt-log1}) for $t\in(0,2]$. This appendix demonstrates the boundedness of $(\lambda^*, \nu^*)$, a key requirement for proving related lemmas and theorems.
	
		By employing backward induction on the first set of equations in the system~(\ref{nfpe-eqt-log1}), we obtain the following equations for $i\in N$, $j\in M_i$, and $a\in A(I^j_i)$,
		\begin{equation}\label{nfpe-eqt-app1}
			\begin{array}{l}
				-\nu^i_{I^j_i}+\sum\limits_{\varpi^i\in W^i,\varpi^i_{I^j_i}a\subseteq \varpi^i}((1-c(t))g^i(\varpi^i,\gamma^{*-i}(t))\\
				\hspace{2.5cm}+\lambda^i(\varpi^i))\prod\limits_{a_q\in \varpi^i\backslash \varpi^i_{I^j_i}a,a_q\in A(I^{j_q}_i)}(1-\theta(t))\beta^i_{I^{j_q}_i}(a_q)=0,
			\end{array}
		\end{equation}
		where \begin{small}\[\beta^i_{I^{j_q}_i}(a_q)=\frac{\gamma^i(t;\varpi^i_{I^{j_q}_i}a_q)-\rho(t)(1-\theta(t))\eta^{0i}(\varpi^i_{I^{j_q}_i}a_q)}{(1-\theta(t))\gamma^i(t;\varpi^i_{I^{j_q}_i})+\theta(t)\gamma^{0i}(\varpi^i_{I^{j_q}_i})-\rho(t)(1-\theta(t))\eta^{0i}(\varpi^i_{I^{j_q}_i})}>0.\]\end{small}It can be seen that $\sum_{a_q\in A(I^{j_q}_i)}\beta^i_{I^{j_q}_i}(a_q)=1$. Specifically, for $(j,a)\in D_i$, the equations (\ref{nfpe-eqt-app1}) follow directly from the first set of equations in the system (\ref{nfpe-eqt-log1}). When $(j,a)\notin D_i$, we assume that the equations (\ref{nfpe-eqt-app1}) hold for all $j_q\in M_i(\varpi^i_{I^j_i}a)$ and $a_q\in A(I^{j_q}_i)$. By multiplying both sides of the equations (\ref{nfpe-eqt-app1}) by $\beta^i_{I^{j_q}_i}(a_q)$ and summing over $a_q\in A(I^{j_q}_i)$, we obtain the expression for $\nu^i_{I^{j_q}_i}$. Finally, by substituting $\zeta^i_{I^j_i}(a)$ with the recursive outcomes, the resulting equation~(\ref{nfpe-eqt-app1}) for $(j,a)\notin D_i$ is derived. For further analysis, we multiply $\beta^i_{I^{j}_i}(a)$ on both sides of the equation (\ref{nfpe-eqt-app1}) and sum over $a\in A(I^j_i)$, yielding
		\begin{equation}\label{nfpe-eqt-app2}
			\begin{array}{l}
				-\nu^i_{I^j_i}+\sum\limits_{a\in A(I^j_i)}\sum\limits_{\varpi^i\in W^i,\varpi^i_{I^j_i}a\subseteq \varpi^i}((1-c(t))g^i(\varpi^i,\gamma^{*-i}(t))\\
				\hspace{2.5cm}+\lambda^i(\varpi^i))\beta^i_{I^{j}_i}(a)\prod\limits_{a_q\in \varpi^i\backslash \varpi^i_{I^j_i}a,a_q\in A(I^{j_q}_i)}(1-\theta(t))\beta^i_{I^{j_q}_i}(a_q)=0.
			\end{array}
		\end{equation}
		Let $L^i_0=\min_{h\in Z}u^i(h)$, $U^i_0=\max_{h\in Z}u^i(h)$ and $Y^i_0=\max_{j\in M_i,a\in A(I^j_i)}\gamma^{0i}(\varpi^i_{I^j_i}a)/(\gamma^{0i}(\varpi^i_{I^j_i})-\eta^{0i}(\varpi^i_{I^j_i}))$ for $i\in N$. The equations (\ref{nfpe-eqt-app2}) indicate that $\nu^i_{I^j_i}\geq-|L^i_0|$ for any $i\in N,j\in M_i$. Then we proceed to analyze the upper bound of $(\lambda^*,\nu^*)$ under two distinct cases.
		
		\textbf{Case 1 ($\frac{3}{2}\leq t\leq 2$):} In this case, the inequality $\frac{1}{3} \leq \theta(t) \leq 1$ holds. From the equations (\ref{nfpe-eqt-app2}) and the third set of equations in the system (\ref{nfpe-eqt-log1}), it follows that $\nu^i_{I^j_i} \leq |U^i_0| + 3|W^i|Y^i_0$ for all $i \in N$ and $j \in M_i$. This result further implies, based on the first set of equations in the system (\ref{nfpe-eqt-log1}), that $\lambda^i(\varpi^i_{I^j_i}a) \leq |W^i||U^i_0| + 3|W^i|Y^i_0 + |L^i_0| + |M_i(\varpi^i_{I^j_i}a)||L^i_0|$ for all $i \in N$, $j \in M_i$, and $a \in A(I^j_i)$.
		
		\textbf{Case 2 ($0< t\leq \frac{3}{2}$):} In this case, the inequality $0\leq \theta(t) \leq \frac{1}{3}$ holds. Consider $i \in N$ and $j \in M_i$ such that $\varpi^i_{I^j_i} = \emptyset$. Since $\lambda^i(\varpi^i)\beta^i_{I^{j}_i}(a)\prod_{a_q\in \varpi^i\backslash \varpi^i_{I^j_i}a,a_q\in A(I^{j_q}_i)}(1-\theta(t))\beta^i_{I^{j_q}_i}(a_q)\leq |W^i|Y^i_0$, it can be drawn from the equations~(\ref{nfpe-eqt-app2}) that $\nu^i_{I^j_i}\leq|U^i_0|+|W^i|Y^i_0\triangleq V^j_i$. Furthermore, $(1-\theta(t))\zeta^i_{I^j_i}(a)\leq V^j_i+|L^i_0|$ holds based on the first set of equations in the system (\ref{nfpe-eqt-log1}). As $\theta(t)\leq\frac{1}{3}$, it follows that $\nu^i_{I^{j_q}_i}\leq \frac{3}{2}(V^j_i+|L^i_0|)+(|M_i(\varpi^i_{I^j_i}a)|-1)|L^i_0|\triangleq V^{j_q}_i$ for $a\in A(I^j_i)$ and $j_q\in M_i(\varpi^i_{I^j_i}a)$. Proceeding by forward induction and noting the game's finiteness, $\nu^i_{I^j_i}$ is bounded above by $V^i_0 = \max_{j_q \in M_i} V^{j_q}_i$ for any $i\in N,j \in M_i$. Finally, the first set of equations in the system (\ref{nfpe-eqt-log1}) implies that $\lambda^i(\varpi^i_{I^j_i}a) \leq V^i_0 + |L^i_0| + |M_i(\varpi^i_{I^j_i}a)||L^i_0|$ for $i\in N,j\in M_i$ and $a\in A(I^j_i)$.
		\section{}\label{nfpe-appen-1}
		This appendix demonstrates that the Jacobian matrix $Dp(x,t,\nu;\alpha)$ of $p(x,t,\nu;\alpha)$ has full-row rank on the domain $\mathbb{R}^{n_0}\times(0,2)\times\mathbb{R}^{m_0}\times\mathbb{R}^{n_0}$, and that $Dp_\alpha(x,2,\nu)$ maintains full-row rank on $\mathbb{R}^{n_0}\times\mathbb{R}^{m_0}$, a property essential for the proof of Theorem~\ref{nfpe-thm-log4}.\par
		Defining $g(x,t,\nu;\alpha)$ as the initial $n_0$ components of $p(x,t,\nu;\alpha)$, we write the Jacobian matrix $Dp(x,t,\nu;\alpha)$ as
		\begin{equation}
			Dp(x,t,\nu;\alpha)=\left(\begin{array}{cccc}
				D_{x} g&D_t g&D_\nu g&-c(t)(2-\theta(t))I^{n_0\times n_0} \\
				B_1+B_2&C&0&0\\
			\end{array}\right),
			\nonumber
		\end{equation}
		where $I^{n_0\times n_0}$ denotes the identity matrix,
		\begin{equation}
			B_1=\left(\begin{array}{cccc}
				{b^1_1}^\top&&&\\
				&{b^2_1}^\top&&\\
				&&\ddots&\\
				&&&{b^{m_n}_n}^\top
			\end{array}\right)\in \mathbb{R}^{m_0\times n_0} 
			\nonumber
		\end{equation}
		with
		\[b^{j}_i=(\frac{d}{dx^i(\varpi^i_{I^j_i}a)}\gamma^i(x,t;\varpi^i_{I^j_i}a):a\in A(I^j_i))^\top\in\mathbb{R}^{|A(I^j_i)|}.\]
		The matrix $B_2\in \mathbb{R}^{m_0\times n_0}$ assigns to each element, where both row and column indices correspond to $\varpi^i_{I^j_i}\neq\emptyset,i\in N,j\in M_i$, the value $(\theta(t)-1)d\gamma^i(x,t;\varpi^i_{I^j_i})/dx^i(\varpi^i_{I^j_i})$, with remaining entries zero. The vector $C\in\mathbb{R}^{m_0}$ represents the partial derivative of the last $m_0$ components of $p(x,t,\nu;\alpha)$ with respect to $t$. Since both $I^{n_0\times n_0}$ and $B_1+B_2$ are of full-row rank, the Jacobian matrix $Dp(x,t,\nu;\alpha)$ has full-row rank on $\mathbb{R}^{n_0}\times(0,2)\times\mathbb{R}^{m_0}\times\mathbb{R}^{n_0}$.\par
		When $t=2$, the system~(\ref{nfpe-eqt-log4}) is reduced into
		\begin{equation*}\label{nfpe-eqt-apn1}\begin{array}{l}
				\lambda^i(x,2;\varpi^i_{I^j_i}a)-\nu^i_{I^j_i}=0,\;i\in N,j\in M_i,a\in A(I^j_i),\\
				\sum\limits_{a\in A(I^j_i)}\gamma^i(x,2;\varpi^i_{I^j_i}a)-\gamma^{0i}(\varpi^i_{I^j_i})=0, i\in N,j\in M_i.\\
			\end{array}
		\end{equation*}
		The Jacobian matrix then takes the form
		\begin{equation}
			Dp_\alpha(x,2,\nu)=\left(\begin{array}{cc}
				F&-E\\
				B_1&0\\
			\end{array}\right),
			\nonumber
		\end{equation}
		where
		\begin{equation}
			E=\left(\begin{array}{cccc}
				{e^1_1}&&&\\
				&{e^2_1}&&\\
				&&\ddots&\\
				&&&{e^{m_n}_n}
			\end{array}\right)\in \mathbb{R}^{n_0\times m_0} \text{ with } e^{j}_i=(1,1,\ldots,1)^\top\in\mathbb{R}^{|A(I^j_i)|},
			\nonumber
		\end{equation}
		and $F=\text{diag}(d\lambda^i(x,2;\varpi^i_{I^j_i}a)/dx^i(\varpi^i_{I^j_i}a): i\in N,j\in M_i,a\in A(I^j_i))$. By applying row and column operations to $Dp_\alpha(x,2,\nu)$, we obtain
		\begin{equation}
			Dp_\alpha(x,2,\nu)=\left(\begin{array}{cc}
				F&-E\\
				0&B_1F^{-1}E\\
			\end{array}\right).
			\nonumber
		\end{equation}		
		Since both $F$ and $B_1F^{-1}E$ are of full-row rank, it follows that $Dp_\alpha(x,2,\nu)$ retains full-row rank on $\mathbb{R}^{n_0}\times\mathbb{R}^{m_0}$.
	\end{appendices}
	
	\newpage
	\bibliography{library}
\end{document}